\documentclass[a4paper]{article}
\usepackage{color}
\usepackage{bm}
\usepackage{amsmath,amssymb,amsthm}
\usepackage[dvips]{graphicx}                                
\def\abstract{\hfil{\large{\bf{Abstract}\vspace{-10pt}\\

}}}

\def\jsum#1#2{\displaystyle\sum_{#1}^{#2}}
\def\jint#1#2{\displaystyle\int_{#1}^{#2}}

\def\jfrac#1#2{\displaystyle\frac{\,#1\,^{\mathstrut}}{\,#2\,_{\mathstrut}}}
\def\jlim#1{\displaystyle \lim_{#1}}

\def\ee{{\mathrm{e}}}
\def\Nu{\mathcal{V}}
\def\oline#1{\overline{#1}}

\newenvironment{e}{\noindent \begin{equation}}{\end{equation}\hspace{-4pt}}
\newtheorem{lem}{Lemma}
\newtheorem{thm}{Theorem}
\newtheorem{exa}{Example}

\title{\textbf{Generalizing the Tomboulis-Yaffe Inequality to $SU(N)$ Lattice Gauge Theories 
and General Classical Spin Systems\vspace{40pt}}}
\author{Takuya Kanazawa\thanks{Email: tkanazawa@nt.phys.s.u-tokyo.ac.jp}\vspace{5pt}\\
\textit{Department of Physics, University of Tokyo, Tokyo 113-0033, Japan}}
\begin{document}
\maketitle
\begin{abstract}
We extend the inequality of Tomboulis and Yaffe in $SU(2)$ lattice gauge theory (LGT) 
to $SU(N)$ LGT and to general classical spin systems, by use of reflection positivity. Basically 
the inequalities guarantee that a system in a box that is sufficiently insensitive to 
boundary conditions has a non-zero mass gap. 
We explicitly illustrate the theorem in some solvable models. Strong coupling expansion is then 
utilized to discuss some aspects of the theorem. Finally a conjecture for exact expression to the off-axis mass gap 
of the triangular Ising model is presented. The validity of the conjecture is tested in multiple ways.
\vspace{-400pt}

\hfill TKYNT-08-13
\end{abstract}
\newpage
\section{Introduction}
In this paper we generalize the inequality proved originally by Tomboulis and Yaffe 
in $SU(2)$ gauge theories \cite{Tomboulis-Yaffe} 
to $SU(N)$ gauge theories with general $N$ and also to a wide range of 
classical spin systems. To make this paper readable for those working on spin systems, we give 
in this section an elementary introduction to studies of quark confinement in lattice gauge theories (LGT) 
with an emphasis on interrelations between concepts in spin systems and those in gauge theories.

That a specific kind of defect could be responsible for determining a phase structure of a statistical system 
is appreciated as a quite useful idea in wide areas of modern physics. It has a long history, possibly dating back to 
R. Peierls' argument on the Ising model \cite{Peierls}. 
In as early as 1944 L. Onsager, in his famous paper on the exact solution of the two-dimensional square Ising model, 
calculated what he called a `boundary tension' (the free energy per unit length of a domain wall separating 
two regions of opposite magnetic order) and found that it is zero above and nonzero 
below the critical temperature \cite{Onsager}; hence 
magnetization is not the only quantity that can characterize the phase structure of the system. 
The history after Onsager clearly tells us the significance of understanding how defects, or dislocations, induce 
the volatility of the order parameter: 
for instance the seminal work of Kosterlitz and Thouless \cite{KT} made understanding 
the infinitely smooth phase transition of the $XY$-model possible by adopting the chemical potential of vortices 
as an order parameter.



Such an idea was imported into the studies of $SU(N)$ gauge theories 
ingeniously by 't Hooft \cite{tHooft}, Mack and Petkova \cite{Mack-Petkova} and several others \cite{Yoneya_etc}. 
Remember that in spin models, the system is said to be in a disordered phase if the two-point 
correlation function decays exponentially making the correlation 
length finite; otherwise the system is said to be either in an ordered phase or 
in a Kosterlitz-Thouless-type phase. 
In parallel, a non-Abelian gauge theory is said to be in a confining phase if the expectation value of a 
Wilson loop decays exponentially with the area it spans; otherwise the system is said to be either 
in the Higgs phase or in the Coulomb (or massless) phase. 
The $SU(N)$ gauge theory with no matter field has been believed to be 
in a confining phase for entire values of coupling constant ($SU(3)$ 
is of special importance as it is supposed to be the true theory of strong interactions in nature, 
where quarks have never been directly observed in  experiments). 
What is mysterious 
is that a rigorous proof of confinement is still missing in spite of a tremendous amount of work 
dedicated to this issue so far. However, according to 
the scenario(s) pioneered by 't Hooft, Mack, Petkova and others \cite{Yoneya_etc}, 
the rapid decay of a Wilson loop expectation value 
might be attributable to a percolation of center vortices. It is an 
object of co-dimension 2 (thus it is a loop in 2+1 spacetime and a closed surface in 3+1 spacetime). 
A rough explanation of their appearance is as follows: 
in pure $SU(N)$ gauge theory, all the fields belong to the adjoint representation of $SU(N)$, 
so that the actual gauge group is $SU(N)/Z_N$ rather than $SU(N)$. $\Pi_1(SU(N)/Z_N)=Z_N$ means 
that the systems has a line defect associated to each element of $Z_N$, 
which is denominated as a `center vortex', or a `'t Hooft loop'. 
If a vortex associated to $z\in Z_N$ wraps around the Wilson loop, the latter 
is multiplied by a factor $z$. 
Then the appearance of infinitely many center vortices 
piercing the Wilson loop randomly \textit{with no mutual correlation} can efficiently disorder 
the value of the Wilson loop, resulting in an exponential suppression of the expectation value 
for larger loops follows.\footnote{It is interesting to note that 
a similarity of the 't Hooft loop in LGT to the domain wall in 
spin models gets even clearer in the \textit{deconfined} phase at high temperature. The action of $SU(N)$ 
LGT possesses a \textit{global} $Z_N$ symmetry, and the confinement-deconfinement transition is 
conventionally interpreted as its spontaneous breaking \cite{S-Y}. The tension of an interface separating 
different $Z_N$ deconfined vacua is calculated perturbatively and numerically 
from the (dual) string tension of the spatial 't Hooft loop \cite{Bursa-Teper_Forcrand-Noth}.}

Roughly speaking, the formulation of 't Hooft concerns a macroscopically large center vortex 
wrapping around the periodic lattice, 
ensuring its presence by imposing a twisted boundary condition on the lattice. 
This procedure is essentially 
tantamount to imposing an anti-periodic condition to produce a domain wall in the Ising model. 
He presented a convincing argument that \textit{the behavior of the free energy of a large vortex in 
approaching the thermodynamic limit characterizes in which phase the system is in}; 
if it vanishes exponentially, then 
the vortices percolates and the system is in the confining phase. On the other hand 
Mack and Petkova formulated a center vortex contained in a torus of finite diameter 
with a fixed boundary condition on the 
surface, and the presence of a center vortex was ensured by a singular 
gauge transformation operated on the surface. In order to elucidate its intuitive meaning 
to spin theorists, 
we would like comment on the concept of `thickness' of the vortex. It is well known that, 
in the continuum, 
an \textit{infinitely thin} center vortex 
is unphysical in the sense that it is associated with an infinite action. 
For illustration let us consider the XY model on a 
one-dimensional chain of length $L$. Suppose we fix the angle at one end of the chain 
to $\varphi$ and the angle at the 
other end to $\varphi+\theta$. If the angles of spins change smoothly as much as possible 
from one end toward the other, the 
energy cost is easily estimated to be $(\theta/L)^2\times L\sim O(1/L)$ for $L\gg 1$. 
This is in sharp contrast to the 
situation in the Ising model on the same chain, in which a smooth change is impossible, 
thus leading to the energy cost of 
$O(1)$ and making a spontaneous symmetry breaking easier to happen. The lesson we learn in this example is 
that it is generally possible to reduce an energy cost associated with a defect by smoothly changing the variables around it; the energy cost associated with a domain wall 
can be reduced if we give it a finite thickness.\footnote{
Dobrushin and Shlosman elevated this idea to a rigorous proof of the absence of magnetic order in 
two-dimensional ferromagnets with a continuous symmetry \cite{Dobrushin-Shlosman}.} 
What Mack and Petkova achieved is to prove an inequality rigorously, whose intuitive 
interpretation being that \textit{the area-law decay of the Wilson loop expectation value would follow 
if the free energy of such a `thick' vortex decreases sufficiently rapidly when its diameter is increased.} 
A lucid exposition of dynamics of thick vortices in $SU(N)$ 
lattice gauge theories (LGT) can be found in ref.\cite{Yaffe}.

As is well known, a fundamental difference between spin systems and gauge theories is that
the latter has no known \textit{local} order parameter (such as magnetization in the former) 
that can characterize the phases of gauge theories. That is why non-local quantities such as Wilson or 
't Hooft loops have been given a special weight in studies of strong-coupling phenomena such as confinement. 
As a classical reference, we would like to mention ref.\cite{GJK} in which 
physical relevance of defects generated by the twisting procedure, including both center vortices in LGT 
and domain walls in spin systems, and usefulness of using them as a probe for the phase structure 
in computer simulations, are reviewed and discussed from a unified point of view. 
\\

Tomboulis and Yaffe thoroughly investigated $SU(2)$ LGT at finite temperature and rigorously 
proved the absence of confinement at sufficiently high temperature \cite{Tomboulis-Yaffe}. In their study 
they derived a number of inequalities between observables such as the Wilson loop expectation value, 
the 't Hooft loop expectation value, the electric flux expectation value and the Polyakov loop correlator. 
Among others they gave an upper bound of the Wilson loop expectation value 
by a specific function of the center vortex 
free energy (as presented in the next section as theorem \ref{first_thm}). 
It gave a firm foundation to 't Hooft's original argument in the continuum \cite{tHooft}, that 
if in approaching the thermodynamic limit the free energy of a center vortex 
that encircles two of the four periodic 
directions of the lattice vanishes exponentially w.r.t. the cross section 
of the lattice perpendicular to the vortex, 
then the area law behavior of the Wilson loop expectation value would follow. 
Thus it sheds light on dynamics of the 
center vortices in a somewhat different manner from the Mack-Petkova inequality. In this paper we call it the Tomboulis-Yaffe (TY) inequality throughout this paper.

The purpose of this paper is to present a generalization of the TY inequality to 
$SU(N)$ LGT for arbitrary 
$N$ and to general classical spin models. Our result gives a rigorous relation between the effect of 
twisted boundary conditions and the correlation function (Wilson loop) 
in spin models (in LGT) respectively.\footnote{Historically, 
changing of boundary conditions has been utilized in studies of Anderson localization as a method 
for estimating the broadening of the wave function \cite{Lee}. More recently it was utilized in the lattice QCD 
calculation \cite{IDIST} to study charmonium properties in deconfinement phase.}
Among spin models, the $SU(N)\times SU(N)$ principal chiral model (PCM) is of particular interest for 
researches of gauge theory, since it bears a number of similarities to $SU(N)$ gauge theories and 
serves as a good testing ground for techniques in gauge 
theories \cite{Green-Samuel,Polyakov,RCV,DPRV}. The action of $SU(N)$ PCM is given by
\begin{e}
S=\beta \sum_{x}^{}\sum_{\mu=1}^{d}
\textrm{Re\,Tr\,}\{U(x) U^\dagger(x+\hat\mu)\},\ \ \ U\in SU(N),\ x\in \mathbb{Z}^d,\label{S_PCM}
\end{e}
where $\hat\mu$ denotes a unit vector in $x^\mu$-direction. It is quite straightforward to 
extend the original TY inequality for $SU(N)$ LGT to $SU(N)$ PCM, using a natural correspondence 
(site $\leftrightarrow$ link, link $\leftrightarrow$ plaquette,\,\dots) and indeed 
the TY inequality for $SU(2)$ PCM has already appeared in the literature \cite{spin-KT,spin-K}. 
On the other hand, however, it is technically nontrivial how to extend it to other more general spin models. 
Let us take $G_2$ PCM as an example. Since $G_2$ is an exceptional group with \textit{trivial} center, 
we can no longer use a twist by a center of the gauge group, which gives rise to a technical difficulty. 
Furthermore the use of center twist for PCM is not physically motivated; in the case of $SU(N)$ gauge theory, 
the use of center element is mandatory, but in PCM we can use any other element of the symmetry group for twist. 
Thus the generality of our formulation, that does not 
rely on the center of the symmetry group at all, seems to be a fundamental progress.\footnote{
As an aside we note that the inequality of Mack and Petkova for $SU(N)$ LGT 
was generalized to $SU(N)$ PCM by Borisenko and Skala \cite{Borisenko-Skala}.}

This paper is organized as follows. In section 2 we will recapitulate the TY inequality for $SU(2)$ 
and then prove its generalization to $SU(N)$. We will use 
the two-dimensional $SU(N)$ LGT to illustrate our result. In section 3 we will prove a 
generalization of the inequality to general classical spin systems. We will use the one-dimensional 
PCM and the two-dimensional Ising models on square and triangular lattices to illustrate the proved 
inequality. Especially, in section \ref{3-5}, 
we derive a rigorous upper bound of the off-axis correlation length in the triangular 
Ising model, whose exact expression is still unknown, 
and conjecture that it is indeed the exact one. In section \ref{SC_expansion} 
the strong coupling expansion technique is employed 
to shed light on the implication of our theorem, as well as to test the conjecture. 
Section 4 is devoted to the conclusion.

%
%
%
%
%
%
%
%
%
%
%
%

\section{TY inequality in LGT}
\subsection{$N=2$}
Let us recapitulate the TY inequality for $SU(2)$ \cite{Tomboulis-Yaffe}. $\Lambda$ is a $d$-dimensional 
hypercubic lattice of length $L_\mu$ ($\mu=1,\dots,d$) with periodic boundary condition and 
$\Nu$, called ``vortex'', is a stacked set of plaquettes winding around the lattice $\Lambda$ in $d-2$ 
periodic directions.\footnote{$\Nu$ forms a closed loop when $d=3$ and a closed surface (2-torus) when 
$d=4$, on the dual lattice. See fig.\ref{link}.} 
We assume the directions unwrapped by $\Nu$ to be $x^\mu$ and $x^\nu$ ($\mu\not=\nu$).
The ordinary and the ``twisted'' partition functions are given by
\begin{align}
Z_\Lambda&\equiv\jint{}{}\prod_{b}^{}dU_b\,\exp\left(\jfrac{\beta}{2}\jsum{p\subset \Lambda}{}\mathrm{Tr\,}U_p\right),
\label{saisyonoaction}
\\
Z^{(-)}_\Lambda&\equiv
\jint{}{}\prod_{b}^{}dU_b\,\exp\left(\jfrac{\beta}{2}\Bigg[\jsum{p\subset \Nu}{}\mathrm{Tr\,}(-U_p)
+\jsum{p\subset \Lambda\setminus \Nu}{}\mathrm{Tr\,}U_p\Bigg]\right),
\end{align}
where $dU$ is the normalized Haar measure of $SU(2)$ and 
$U_p\equiv U_{x,\mu}U_{x+\mu,\nu}U^\dagger_{x+\nu,\mu}U^\dagger_{x,\nu}$ 
is a plaquette variable. It is important that local redefinition of variables $U\to -U$ can move the locations of twisted 
plaquettes but cannot remove the twist from $\Lambda$ entirely.

Next, consider a rectangle $C$ lying in a $x^\mu$-$x^\nu$ plane with $A_C$ the area enclosed by $C$, and 
let $W(C)$ the Wilson loop in the fundamental representation associated with $C$, namely 
$\displaystyle W(C)\equiv\jfrac{1}{2}{\mathrm{Tr}\,}\prod_{b\in C}^{}U_b$. 
Then the following inequality holds \cite{Tomboulis-Yaffe,Kovacs-Tomboulis-1}:
\begin{thm}\label{first_thm}
\begin{e}
\langle W(C)\rangle\leq 2\left\{\jfrac{1}{2}\left(1-\jfrac{Z_\Lambda^{(-)}}{Z_\Lambda}
\right)\right\}^{A_C/L_\mu L_\nu}   \label{N=2}
\end{e}
where $\langle\dots\rangle$ is the expectation value w.r.t. the measure of $Z_\Lambda$.
\end{thm}

The site-reflection positivity of the Wilson action \cite{Montvay-Munster} 
plays an essential role in the proof. (As is well known, the Wilson action is among those actions for which 
the link-reflection positivity is also satisfied \cite{Osterwalder-Seiler} but it is not a matter of interest here.) 
Indeed (\ref{N=2}) can be proved with any one-plaquette action, since they are site-reflection positive (although 
not necessarily link-reflection positive, of course). 

An important implication of (\ref{N=2}) is that \textit{the area-law decay of $\langle W(C)\rangle$ would follow if 
$1-Z_\Lambda^{(-)}/Z_\Lambda\approx \ee^{-\rho L_\mu L_\nu}$ for some constant $\rho>0$ in the 
thermodynamic limit}\footnote{We neglected the entropy factor 
for simplicity.}. This is a famous criterion of confinement originally proposed by 't Hooft \cite{tHooft} and is also 
numerically supported \cite{Kovacs-Tomboulis,de Forcrand-D'Elia-Pepe,de Forcrand-Smekal}. 
Moreover such a behavior of $Z_\Lambda^{(-)}/Z_\Lambda$ has been verified explicitly by $\rm{M\ddot unster}$ using the 
convergent strong-coupling cluster expansion technique \cite{Munster}. See page \pageref{Munster_} for more discussion 
on this point.

Theorem \ref{first_thm} was utilized in a recent attempt at a rigorous proof of confinement \cite{Tomboulis} 
with related discussions \cite{jibun,I-S}.

\subsection{General $N$}
The authors of ref.\cite{Tomboulis-Yaffe} state without explicit construction that 
their result is extendable to any other gauge group with nontrivial center. 
Since the mentioned extension does not seem to be so trivial and, 
to the author's best knowledge an explicit formula for general $N$ is not found in the 
literature, we think it valuable to present the extension of (\ref{N=2}) from $SU(2)$ to $SU(N)$ together with its proof.

Let us give a formulation of vortices in $SU(N)$ LGT and prove their properties 
before presenting TY inequality in $SU(N)$ LGT. 
The ordinary and the ``twisted'' partition functions are respectively given by
\begin{align}
Z_\Lambda&\equiv\jint{}{}\prod_{b}^{}dU_b\,\exp\left(\jfrac{\beta}{N}\jsum{p\subset \Lambda}{}\mathrm{Re\,Tr\,}U_p\right),
\label{partition__}
\\
Z^{[k]}_\Lambda&\equiv
\jint{}{}\prod_{b}^{}dU_b\,\exp\left(\jfrac{\beta}{N}\Bigg[\jsum{p\subset \Nu}{}\mathrm{Re\,Tr\,}(z^kU_p)
+\jsum{p\subset \Lambda\setminus \Nu}{}\mathrm{Re\,Tr\,}U_p\Bigg]\right),
\end{align}
\begin{e}
z\equiv \exp\Big(\frac{2\pi i}{N}\Big),\ \ k\equiv 1,2,...,N-1\,({\rm mod}\ N).
\end{e}
Hereafter $\langle\dots\rangle$ represents the expectation value with the measure (\ref{partition__}). 
The vortex creation operators $\{\mathcal{O}^{[k]}\}$ and the electric flux creation operators 
$\{\mathcal{F}^{[m]}\}$ are defined by
\begin{gather}
\langle \mathcal O^{[k]}[\Nu]\rangle\equiv\jfrac{Z_\Lambda^{[k]}}{Z_\Lambda},
\\
\langle\mathcal F^{[m]}[\Nu]\rangle\equiv\jfrac{1}{N}\jsum{k=0}{N-1}z^{mk}\langle \mathcal O^{[k]}[\Nu]\rangle
.
\label{F_defdef}
\end{gather}
Thus we have $\langle \mathcal O^{[k]}[\Nu]\rangle=\jsum{m=0}{N-1}z^{-km}\langle \mathcal F^{[m]}[\Nu]\rangle$. 
The explicit form of $\mathcal O$ is given by\vspace{-5pt}
\begin{e}
\mathcal O^{[k]}[\Nu]=\exp\hspace{-2pt}\Big(\jfrac{\beta}{N}\jsum{p\subset \Nu}{}
[\mathrm{Re\,Tr\,}(z^kU_p)-\mathrm{Re\,Tr\,}U_p]\Big).
\end{e}
\begin{lem}
If $\Nu,\,\Nu',\,\Nu'',\dots$ are homologous\footnote{Plural vortices are called \textit{homologous} if and only if 
they wind around the same periodic directions of $\Lambda$.}, we have
\begin{gather}
\langle\mathcal O^{[k]}[\Nu]\mathcal O^{[k']}[\Nu']\rangle=\langle\mathcal O^{[k+k']}[\Nu]\rangle,\label{b1}
\\
\jsum{m=0}{N-1}\mathcal F^{[m]}[\Nu]=1,   \label{b2}
\\
\langle\mathcal F^{[l]}[\Nu] \mathcal F^{[m]}[\Nu']\rangle
=\langle\mathcal F^{[l]}[\Nu]\rangle\,\delta^{(N)}_{l,m},
\label{b3}
\\
\langle\mathcal F^{[l]}[\Nu]\mathcal F^{[m]}[\Nu'] \mathcal F^{[n]}[\Nu'']\rangle
=\langle\mathcal F^{[l]}[\Nu]\rangle\,\delta^{(N)}_{l,m}\delta^{(N)}_{m,n},
\label{b4}
\\
\vdots\notag
\end{gather}
where $\delta^{(N)}_{l,m}=1$ if $l\equiv m$ (mod $N$) and $\delta^{(N)}_{l,m}=0$ otherwise.
\end{lem}
\begin{proof}
(\ref{b1}) can be derived by iterating the redefinition of variables $U\to z^{k'}U$ to bring $\Nu'$ to $\Nu$. 
Relations (\ref{b2})-(\ref{b4}) follow from (\ref{F_defdef}) and (\ref{b1}).
\end{proof}
(\ref{b2})-(\ref{b4}) imply that {\textit{$\{\mathcal F^{[m]}\}_m$ can be seen as projection operators}} \cite{tHooft,Yaffe}.

Consider a $(d-1)$-dimensional hyperplane $\pi$ defined by $x^\mu=m$ with $m\in\mathbb{Z}$ fixed.\footnote{
In this paper we never use hyperplanes defined by $x^\mu=m+\frac{1}{2}$; that is, we never use link-reflections.} 
Define the reflection operator $\theta$ w.r.t. $\pi$ by 
$\theta\big[F(\{U_b\})\big]=\overline{F(\{U_{\theta[b]}\})}$ where $F$ is an arbitrary observable (that is, a 
map from configurations on $\Lambda$ to $\mathbb{C}$). The reflection $\theta[b]$ of a link $b$ is also defined by the 
same notation where locations of $b$ and $\theta[b]$, are defined to be symmetrical about $\pi$.
\begin{lem}With $\Nu^\theta\equiv \theta[\Nu]$ we have
\begin{gather}
\theta\Big[\mathcal O^{[k]}[\Nu]\Big]=\mathcal O^{[-k]}[\Nu^\theta],\label{a1}
\\
\theta\Big[\mathcal F^{[m]}[\Nu]\Big]=\mathcal F^{[m]}[\Nu^\theta],\label{a2}
\\
0\leq\langle \mathcal O^{[k]}[\Nu]\rangle\leq 1,\label{a3}
\\
0\leq\langle \mathcal F^{[m]}[\Nu]\rangle\leq 1.\label{a4}
\end{gather}
\end{lem}
\begin{proof}
(\ref{a1}) is obvious from the fact that the orientation of plaquettes are reversed by reflection. (\ref{a1}) yields
\begin{align}
\theta\Big[\mathcal F^{[m]}[\Nu]\Big]&=\theta\Big[\jfrac{1}{N}\jsum{k=0}{N-1}z^{mk}\mathcal O^{[k]}
[\Nu]\Big]\\
&=\jfrac{1}{N}\jsum{k=0}{N-1}z^{-mk}\mathcal{O}^{[-k]}[\Nu^\theta]
\\
&=\mathcal F^{[m]}[\Nu^\theta]
\end{align}
which proves (\ref{a2}). Next, using the Schwarz inequality $|\langle F\rangle|\leq\langle F\theta F\rangle^{1/2}$ 
and (\ref{b1}),\,(\ref{a1}) we find
\begin{align}
\langle \mathcal O^{[k]}[\Nu]\rangle&\leq\langle \mathcal O^{[k]}[\Nu]\theta\Big[\mathcal O^{[k]}[\Nu]\Big]\rangle^{1/2}
\\
&=\langle\mathcal O^{[k]}[\Nu]\mathcal O^{[-k]}[\Nu^\theta]\rangle^{1/2}=1.
\end{align}
which proves the second inequality in (\ref{a3}) while the first one is trivial. 
Since $\Nu$ and $\Nu^\theta$ are homologous we can apply (\ref{b3}) to obtain
\begin{align}
\langle\mathcal F^{[m]}[\Nu]\rangle&=\langle\mathcal F^{[m]}[\Nu]\mathcal F^{[m]}[\Nu^\theta]\rangle
\\
&=\langle\mathcal F^{[m]}[\Nu]\theta\Big[\mathcal F^{[m]}[\Nu]\Big]\rangle\geq 0.  \label{sisi}
\end{align}
(\ref{sisi}) and (\ref{b2}) prove (\ref{a4}). (These simple proofs of (\ref{a3}) and (\ref{a4}) seem to be new.)
\end{proof}
The vortex free energy $F_v^{[k]}$ and the electric flux free energy $F_{el}^{[m]}$ are 
defined by $\ee^{-F_v^{[k]}}\equiv\langle\mathcal O^{[k]}[\Nu]\rangle$  
and $\ee^{-F_{el}^{[m]}}\equiv\langle\mathcal F^{[m]}[\Nu]\rangle$, respectively.
\\

Let $N(R)\in\{0,1,\dots,N-1\}$ denote the $N$-ality of an irreducible representation 
$R$\footnote{$N$-ality is the number (mod $N$) of boxes in the Young tableau of $R$.} of $SU(N)$ whose dimension is $d_R$. 
Take a rectangle $C$ lying in a $x^\mu$-$x^\nu$ plane and let $A_C$ the area enclosed by $C$. For the 
normalized Wilson loop in the representation $R$, 
$W_R(C)\equiv\jfrac{1}{d_R}\chi_R\Big(\prod_{b\in C}^{}U_b\Big)$, we have
\begin{thm}[TY inequality for $SU(N)$ LGT]\label{thm_N}
\begin{e}
|\langle W_R(C)\rangle|\leq \langle \mathcal F^{[N(R)]}[\Nu]\rangle^{A_C/L_\mu L_\nu}
+\left\{1-\langle \mathcal F^{[0]}[\Nu]\rangle\right\}^{A_C/L_\mu L_\nu}.                        \label{N}
\end{e}
In addition, if $N(R)\not =0$ we have
\begin{align}
|\langle W_R(C)\rangle|&\leq 2\left\{1-\langle \mathcal F^{[0]}[\Nu]\rangle\right\}^{A_C/L_\mu L_\nu}   \label{local__}
\\
&=2\left\{1-\jfrac{1}{N}\jsum{k=0}{N-1}\langle\mathcal O^{[k]}[\Nu]\rangle\right\}^{A_C/L_\mu L_\nu}. 
\label{N_cor}
\end{align}
\end{thm} 
\begin{proof}
Although the argument below parallels that of 
ref.\cite{Kovacs-Tomboulis-1} for $SU(2)$, we describe the proof in full 
detail for readers' convenience. Suppose $\Nu,\,\Nu'$ are stacked set of plaquettes wrapping 
around $d-2$ periodic directions of $\Lambda$ and $\Nu$ is linking once with $C$ while $\Nu'$ is not. 
(See fig.\ref{link} for a 3-dimensional illustration of the setting.)\vspace{-10pt}
\begin{figure}[hbtp]
\begin{center}
\includegraphics[width=6cm,clip]{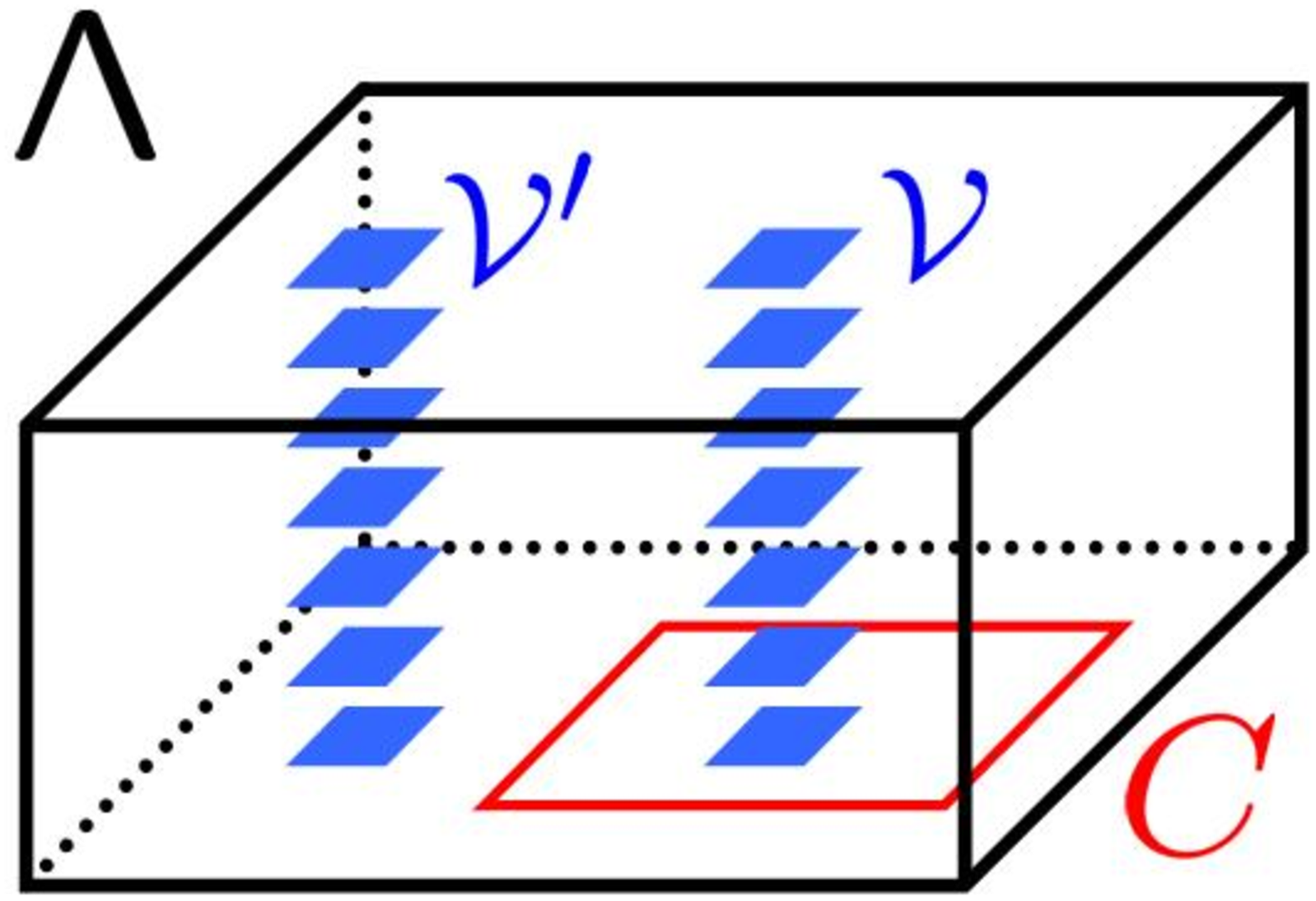}
\end{center}\vspace{-16pt}
\caption{Locations of $\Nu,\,\Nu'$ and $C$.}
\label{link}
\end{figure}

Let us rewrite the expectation value of 
$W_R(C)\equiv\jfrac{1}{d_R}\chi_R\Big(\displaystyle\prod_{b\in C}U_b\Big)$ as follows:
\begin{align}
\langle W_R(C)\rangle&=\langle(1-\mathcal F^{[0]}[\Nu])W_R(C)\rangle+\langle\mathcal F^{[0]}[\Nu]W_R(C)\rangle
\\
&=\langle(1-\mathcal F^{[0]}[\Nu])W_R(C)\rangle+\jfrac{1}{N}\jsum{k=0}{N-1}
\langle\mathcal O^{[k]}[\Nu]W_R(C)\rangle.                                    \label{ii7}
\end{align}
The presence of the second term is not desirable from the viewpoint of 
obtaining a meaningful upper bound of $\langle W_R(C)\rangle$, 
so let us perform redefinitions of variables $U\to z^k U$ 
to bring $\Nu$ to $\Nu'$, which causes the change
\begin{e}
\langle\mathcal O^{[k]}[\Nu]W_R(C)\rangle\to z^{\pm N(R)k}\langle\mathcal O^{[k]}[\Nu']W_R(C)\rangle.
\end{e}
This is because in the course of bringing $\Nu$ to $\Nu'$ we must change one of the link variables on $C$. 
(The sign of exponent depends on the orientation of $C$.) Thus (\ref{ii7}) becomes
\begin{align}
\langle W_R(C)\rangle&=\langle(1-\mathcal F^{[0]}[\Nu])W_R(C)\rangle+\jfrac{1}{N}\jsum{k=0}{N-1}
z^{\pm N(R)k}\langle\mathcal O^{[k]}[\Nu']W_R(C)\rangle
\\
&=\langle(1-\mathcal F^{[0]}[\Nu])W_R(C)\rangle+\langle\mathcal F^{[\pm N(R)]}[\Nu']W_R(C)\rangle.
\end{align}
\begin{figure}[hbtp]
\begin{center}
\includegraphics[width=15.0cm,clip]{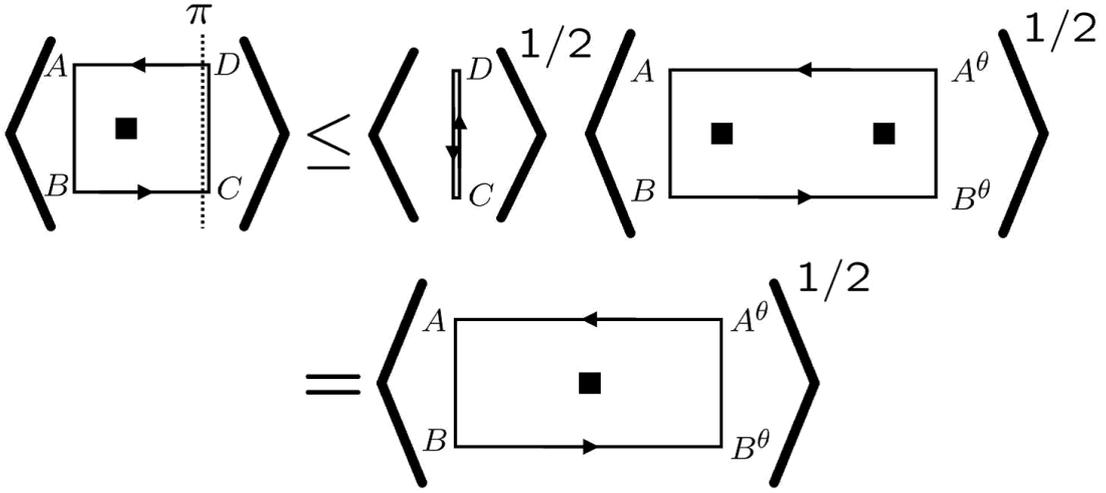}
\end{center}\vspace{-15pt}
\caption{The Schwarz inequality enables us to double the size of the rectangle. This figure is essentially 
borrowed from ref.\cite{Kovacs-Tomboulis-1}.}
\label{kaworu}
\end{figure}
Our next step is expressed in fig.\ref{kaworu} schematically in which a black square represents the operator 
$(1-\mathcal F^{[0]}[\Nu])$. Labeling four vertices as $A,B,C,D$, setting the hyperplane $\pi$ so that it is 
perpendicular to the rectangle and contains the edge $CD$, and applying the Schwarz inequality 
$|\langle F\theta G\rangle|\leq\langle F\theta F\rangle^{1/2}\langle G\theta G\rangle^{1/2}$ 
we obtain
\begin{align}
&\big|\langle(1-\mathcal F^{[0]}[\Nu])W_R\big|_{DA+AB+BC+CD}\rangle\big|
\\
=\ &\jfrac{1}{d_R}\big|\langle(1-\mathcal F^{[0]}[\Nu])\chi_R\big|_{DA+AB+BC+CD}\rangle\big|
\\
=\ &\jfrac{1}{d_R}\Big|\jsum{\alpha,\beta}{}\big\langle(1-\mathcal F^{[0]}[\Nu])\{\chi_R\big|_{DA+AB+BC}\}_{\alpha\beta}
\{\chi_R\big|_{CD}\}_{\beta\alpha}\big\rangle\Big|
\\
\leq \ &\jfrac{1}{d_R}\jsum{\alpha,\beta}{}\Big\langle(1-\mathcal F^{[0]}[\Nu])(1-\mathcal F^{[0]}[\Nu^\theta])
\{\chi_R\big|_{DA+AB+BC}\}_{\alpha\beta}
\overline{\{\chi_R\big|_{DA^\theta+A^\theta B^\theta+B^\theta C}\}_{\alpha\beta}}\Big\rangle^{1/2}
\notag
\\
&\hspace{100pt}\times 
\Big\langle \{\chi_R\big|_{CD}\}_{\beta\alpha}\overline{\{\chi_R\big|_{CD}\}_{\beta\alpha}}\Big\rangle^{1/2}     \label{x1}
\\
\leq \ &\jfrac{1}{d_R}\Big\langle(1-\mathcal F^{[0]}[\Nu])\jsum{\alpha,\beta}{}\{\chi_R\big|_{DA+AB+BC}\}_{\alpha\beta}
\overline{\{\chi_R\big|_{DA^\theta+A^\theta B^\theta+B^\theta C}\}_{\alpha\beta}}\Big\rangle^{1/2}
\times\sqrt{d_R}
\label{x2}
\\
=\ &\langle(1-\mathcal F^{[0]}[\Nu])W_R\big|_{AB+BB^\theta+B^\theta A^\theta+A^\theta A}\rangle^{1/2}.
\end{align}
In going from (\ref{x1}) to (\ref{x2}) we used (\ref{b3}). The length of the rectangle doubled.

Let $l_\mu\times l_\nu$ the size of the original rectangle and assume that 
$L_\mu=2^pl_\mu$ and $L_\nu=2^ql_\nu$ for some $p,q\in\mathbf{N}$\footnote{This condition was also present 
in the original TY inequality \cite{Tomboulis-Yaffe}. 
It is not a severe restriction, however, as long as we believe the asymptotic behavior of observables to be 
independent of the way we take $L_\mu,\,L_\nu$ to infinity.}. Repeating the operation above for sufficiently many times in 
both $x^\mu$- and $x^\nu$- directions, due to the periodic boundary conditions the rectangle finally 
vanishes away, yielding
\begin{align}
\langle(1-\mathcal F^{[0]}[\Nu])W_R(C)\rangle&\leq\left\{1-\langle\mathcal F^{[0]}[\Nu]\rangle\right\}^{1/2^{p+q}}
\\
&=\left\{1-\langle\mathcal F^{[0]}[\Nu]\rangle\right\}^{l_\mu l_\nu/L_\mu L_\nu}
=\left\{1-\langle\mathcal F^{[0]}[\Nu]\rangle\right\}^{A_C/L_\mu L_\nu}.
\end{align}
$\langle\mathcal F^{[\pm N(R)]}[\Nu']W_R(C)\rangle
\leq \langle\mathcal F^{[\pm N(R)]}[\Nu']\rangle^{A_C/L_\mu L_\nu}$ can be shown in a similar way, 
hence (\ref{N}) is proved.

(\ref{local__}) is a consequence of (\ref{N}) and $\langle\mathcal F^{[N(R)]}[\Nu]\rangle
\leq 1-\langle \mathcal F^{[0]}[\Nu]\rangle$.
\end{proof}
The message of (\ref{N_cor}) is that \textit{the exponential decay of the vortex free energy, i.e. 
$\langle \mathcal O^{[k]}[\Nu]\rangle\equiv \ee^{-F_v^{[k]}}=1-O(\ee^{-\rho L_\mu L_\nu})$ for \underline{every $k$}, 
is a sufficient condition for the area law of the Wilson loop to hold}. The area law does not hold, or is 
at least difficult to prove, if not all of the $\langle\mathcal O\rangle$'s converge to 1. 

Several comments are in order. Firstly, suppose that 
the action in (\ref{saisyonoaction}) is in the \textit{adjoint} representation. 
Then $\langle\mathcal O^{[k]}[\Nu]\rangle=1$ follows for any $k$, hence making $\langle W_R(C)\rangle$ with $N(R)\not=0$ 
vanish identically for arbitrary finite volume (see (\ref{N_cor})). This is to be anticipated; 
since the adjoint Wilson action is invariant under 
the \textit{local} transformation $U\to zU$ $(z\in Z_N)$, which cannot break spontaneously according to the 
Elitzur's theorem, and since $W_R(C)$ with $N(R)\not =0$ transforms nontrivially under this transformation, 
its expectation value {\textit{must}} vanish.

Secondly, there are {\it{a lot more varieties}} of inequalities available other than (\ref{N}). 
Assume $N=6$ for instance. From (\ref{b3}) we have 
$\big\langle(\mathcal F^{[1]}[\Nu]+\mathcal F^{[2]}[\Nu])(\mathcal F^{[1]}[\Nu']+\mathcal F^{[2]}[\Nu'])
\big\rangle=\langle\mathcal F^{[1]}[\Nu]+\mathcal F^{[2]}[\Nu]\rangle$ and 
$\big\langle(\mathcal F^{[3]}[\Nu]+\mathcal F^{[4]}[\Nu]+\mathcal F^{[5]}[\Nu])
(\mathcal F^{[3]}[\Nu']+\mathcal F^{[4]}[\Nu']+\mathcal F^{[5]}[\Nu'])\big\rangle
=\langle\mathcal F^{[3]}[\Nu]+\mathcal F^{[4]}[\Nu]+\mathcal F^{[5]}[\Nu]\rangle$, hence by
modifying the above proof one can straightforwardly show
\begin{e}
|\langle W_R(C)\rangle|\leq \langle \mathcal F^{[N(R)]}\rangle^{A_C/L_\mu L_\nu}
+\left\{\langle\mathcal F^{[1]}+\mathcal F^{[2]}\rangle\right\}^{A_C/L_\mu L_\nu}
+\left\{\langle\mathcal F^{[3]}+\mathcal F^{[4]}+\mathcal F^{[5]}\rangle\right\}^{A_C/L_\mu L_\nu}. \label{N=6}
\end{e}
However (\ref{N=6}) and all of its cousins are weaker than (\ref{N}) with $N=6$, which can be understood by 
the elementary inequality $(\sum_{i}^{}x_i)^\alpha<\sum_{i}^{}(x_i)^\alpha$ for $0<\alpha<1$ and $0<x_i$.

Note that one cannot derive the area law from (\ref{N}) when $N(R)=0$, as can be seen from
\begin{e}
\Big[{\rm{r.h.s.}}\ {\rm{of}}\ (\ref{N})\Big]\geq \langle\mathcal F^{[0]}[\Nu]\rangle^{A_C/L_\mu L_\nu}
\geq \Big(\jfrac{1}{N}\Big)^{A_C/L_\mu L_\nu}\ \ \to 1\ \ {\rm{as}}\ \ L_\mu,\,L_\nu\to\infty.
\end{e}
The above implies that the ``gluons'' of $SU(N)$ can screen particles of zero $N$-ality.

Thirdly, theorem \ref{thm_N} is correct even after a 
matter field whose $N$-ality is zero is introduced into the 
theory, since the matter-gauge coupling $\Phi^\dagger_{x+\hat\mu} D_r[U_{x,\mu}]\Phi_x$ preserves reflection positivity and 
is insensitive to the change of variables $U\to zU$.

Finally we remark on the utility of strong-coupling cluster expansion techniques. (Similar discussion will 
be presented in section \ref{SC_expansion}.) \label{Munster_}As already mentioned, 
the exponential suppression of vortex free energy 
$-\log\langle \mathcal O[\Nu]\rangle\approx \ee^{-\rho L_\mu L_\nu}$
has been verified by $\rm{M\ddot unster}$ \cite{Munster} for $SU(2)$ LGT and for sufficiently strong coupling. 
Especially he showed to all orders of strong-coupling expansion that the constant $\rho$ appearing in the vortex 
free energy ('t Hooft's string tension) is equal to the conventional Wilson's string tension. 
His proof hinges on the observation that both the calculation 
of Wilson loop expectation value and that of vortex free energy reduce, at sufficiently strong coupling, to the problem 
of fluctuating \textit{random surfaces}. 
It is understood without difficulty that the methods he employed can be readily used for $SU(N)$ LGT to show 
$-\log\langle \mathcal O^{[k]}[\Nu]\rangle\approx\ee^{-\rho L_\mu L_\nu}$ for every $k\not=0$ (hence proving the area law). 
This time $\rho$ is equal to the \textit{fundamental} string tension (since the gauge action 
(\ref{partition__}) is in the fundamental representation).
\vspace{5pt}
\\
\hfil\hspace{80pt}*\hspace{80pt}*\hspace{80pt}*\hspace{80pt}*\hspace{80pt}

Let us then turn to LGT with matter field of \textit{non-zero $N$-ality}; 
the relations (\ref{b1})-(\ref{b4}) are no longer valid. 
If the matter field has $N$-ality $m$ and the greatest common divisor of $N$ and $m$ is $s$ , the subgroup
$Z_s\subset Z_N$ is a symmetry of the theory. It is thus straightforward to prove the following
\begin{thm}If $N(R)\not \equiv 0$ (mod $s$), we have
\begin{e}
|\langle W_R(C)\rangle|\leq 2\left\{1-\langle f^{[0]}[\Nu]\rangle\right\}^{A_C/L_\mu L_\nu}
=2\left\{1-\jfrac{1}{s}\jsum{k=0}{s-1}\jfrac{Z_\Lambda^{[kN/s]}}{Z_\Lambda}\right\}^{A_C/L_\mu L_\nu}.
\end{e}
\end{thm}
Although the development so far has been for $SU(N)$ LGT, the inequalities evidently apply to $U(N)$ LGT since the center 
of $U(N)$ is $U(1)$ which contains all of $Z_2,\,Z_3,\,Z_4,\,Z_5,\,\dots$. Let us focus on $U(1)$ for simplicity and 
define the twisted partition function as
\begin{align}
Z_\Lambda(\theta)&\equiv
\jint{}{}\prod_{b}^{}dU_b\,\exp\left(\beta\Bigg[\jsum{p\subset \Nu}{}\mathrm{Re\,}(e^{i\theta}U_p)
+\jsum{p\subset \Lambda\setminus \Nu}{}\mathrm{Re\,}U_p\Bigg]\right).
\end{align}
We state below the counterpart of theorem \ref{thm_N}. The proof is straightforward.
\begin{thm}For the Wilson loop of $U(1)$-charge $q\in\mathbb{Z}\setminus\{0\}$, we have
\begin{e}\hspace{25pt}
|\langle W_q(C)\rangle|
\leq 2 \left\{1-\big\langle \mathcal F_{U(1)}^{[0]}[\Nu]\big\rangle\right\}^{A_C/L_\mu L_\nu}
=2\left\{\jint{0}{2\pi}\jfrac{d\theta}{2\pi}\Bigg(1-\jfrac{Z_\Lambda(\theta)}{Z_\Lambda}\Bigg)\right\}^{A_C/L_\mu L_\nu},
\end{e}
with
\begin{e}
\langle\mathcal F_{U(1)}^{[0]}[\Nu]\rangle\equiv\jint{0}{2\pi}\jfrac{d\theta}{2\pi}
\jfrac{Z_\Lambda(\theta)}{Z_\Lambda}.
\end{e}
\end{thm}
Finally we point out that theorem \ref{thm_N} can be proved even if we add a Wilson loop with \textbf{zero $\bm{N}$-ality}, 
of size $2\times 2$ or $2\times 1$ to the action. It is simply because the site-reflection positivity is kept and 
the algebras of twists are still well defined. 

\subsection{Demonstration in 2D $SU(N)$ LGT}
Let us explicitly verify the proved inequality in solvable two-dimensional $SU(N)$ LGT. 
In two dimension, the twist is introduced on just one plaquette. Let the size of the lattice 
$L_1\times L_2$ and impose periodic boundary conditions in both directions. 
First we expand the exponentiated one-plaquette action, $\ee^{-S_p}$, into characters of 
irreducible unitary representations of $SU(N)$:
\begin{e}
e^{-S_p(U)}=\jsum{r}{}d_rF_r\chi_r(U)
\end{e}
where $d_r$ denotes the dimension of a representation $r$ and 
the reality of $S_p$ implies $F_r=F_{\overline{r}}$ (overline 
represents complex conjugation). Then a 
straightforward calculation using formulae
\begin{e}
\jint{}{}dU\,\chi_r(VU)\chi_{r'}(U^\dagger W)=\jfrac{1}{d_r}\delta_{rr'}\chi_r(VW),
\end{e}
\begin{e}
\jint{}{}dU\,\chi_r(VUWU^\dagger)=\jfrac{1}{d_r}\chi_{r}(V)\chi_{r}(W),
\end{e}
yields\footnote{(\ref{NEC}) differs from that obtained in ref.\cite{Gross-Witten} 
because they impose \textit{free} boundary conditions.}
\begin{align}
Z_{\Lambda}&\equiv\jint{}{}\prod_{b\in\Lambda}dU_b\,
\prod_p \left(\jsum{r}{}d_rF_r\chi_r(U_p)\right)
\\
&=\jsum{r}{}(F_r)^{L_1 L_2}.       \label{NEC}
\end{align}
On the other hand, introducing a twist 
$z^k=\exp(2\pi ik/N)$ into arbitrary one 
plaquette on $\Lambda$ gives the twisted partition function
\begin{e}
Z^{[k]}_\Lambda=\jsum{r}{}(F_r)^{L_1 L_2}z^{kN(r)}.
\end{e}
Using the identity $\jsum{k=0}{N-1}z^{kN(r)}=N\delta_{0,N(r)}$ we easily obtain
\begin{align}
1-\jfrac{1}{N}\jsum{k=0}{N-1}\langle\mathcal O^{[k]}\rangle
&=\jfrac{\jsum{r;\,N(r)\not =0}{}(F_r)^{L_1 L_2}}{\jsum{r}{}(F_r)^{L_1 L_2}}
\\
&=\jfrac{\jsum{r;\,N(r)\not =0}{}(c_r)^{L_1 L_2}}
{1+\jsum{r\not =T}{}(c_r)^{L_1 L_2}}
\end{align}
where $T$ implies the trivial representation and we defined 
$c_r\equiv \jfrac{F_r}{F_T}$. Note that $c_r=
c_{\overline{r}}$. Since
\begin{e}
|F_r|=\Big|\jfrac{1}{d_r}\jint{}{}dU\,\ee^{-S_p(U)}\overline{\chi_r(U)}\Big|
<\jint{}{}dU\,\ee^{-S_p(U)}=F_T,
\end{e}
we have $|c_r|<1$, while $0\leq c_r$ can be shown for a 
wide class of gauge actions including the Wilson action.

From above considerations we obtain
\begin{align}
\Big[{\rm{r.h.s.\ of\ }}(\ref{local__})\Big]
&=2\left\{1-\jfrac{1}{N}\jsum{k=0}{N-1}\langle\mathcal O^{[k]}\rangle\right\}^{A_C/L_1L_2}
\\
&\to 2(c_{r'})^{A_C}
\ \ \ \ \ \ \ {\textrm{as}}\ \ \ L_1 L_2\to\infty.
\end{align}
Here $c_{r'}$ is defined as the largest value among $\{c_r\,|\,N(r)\not =0\}$. 
In order to determine $r'$ we need an explicit form of the action $S_p$.

Let us turn to the l.h.s. of (\ref{local__}), i.e. the 
Wilson loop expectation value. We 
borrow the result of ref.\cite{Gross-Witten} which in our 
notation reads
\begin{e}
\langle W_R(C)\rangle=(c_{\overline{R}})^{A_C}.
\end{e}
If $N(R)\not=0$, we obviously have 
$c_{\overline{R}}\leq c_{r'}$, hence the inequality 
(\ref{local__}) holds for sure.

\section{Extension of inequalities to spin systems}\label{7899}
Main result of this section is theorem \ref{thm} 
on page \pageref{thm}, which is a generalization of theorem \ref{thm_N} to general spin systems. 
Before that, we need some preliminary analyses.
\subsection{Basic formulation}\label{Basic formulation}
Let us formulate a twisting procedure in spin systems obeying ref.\cite{GJK}. Consider a statistical system 
with nearest-neighbor interactions whose partition function is given by
\begin{e}
Z_\Lambda\equiv\jint{}{}\prod_{x\in\Lambda}d\phi_x\ \exp\left(
\jsum{y\in\Lambda}{}\jsum{\mu=1}{d}A(\phi_y,\phi_{y+\hat\mu})\right),     \label{partition_}
\end{e}
where $\Lambda$ is a $d$-dimensional hypercubic lattice with periodic boundary conditions, $\hat\mu$ is a 
unit vector in the $\mu$-direction and 
$\jsum{y\in\Lambda}{}\jsum{\mu=1}{d}$ is a sum over all links in $\Lambda$. The real-valued symmetric 
function $A(\,,\,)$ dictates the interaction between nearest sites (and possibly includes self interactions on each site). 
It can be shown by standard arguments that site-reflection positivity is automatically satisfied for any 
nearest neighbor interaction (see p.33 of ref.\cite{Montvay-Munster}) 
while link-reflection, not needed in the following, is often violated. 
Hereafter $\langle\dots\rangle$ represents the expectation value with the measure (\ref{partition_}). 
Let us assume that the system is invariant under a global transformation $\phi\to g\phi$ for any element $g$ of a 
global symmetry group $G$:\footnote{Note that $G$ need not be 
the maximal symmetry group of the system. The development in this section 
still holds if we take as $G$ an arbitrary subgroup of the maximal symmetry group.}
\begin{e}\hspace{80pt}
A(\phi,\phi')=A(g\phi,g\phi'),\hspace{70pt}g\in G.
\end{e}
We assume that $G$ is compact.

A twist for a link is defined as the change of interaction from $A(\phi,\phi')$ to 
$A(\phi,g\phi')$. An important difference from the twist in LGT is that 
\textit{\,$g$ need not belong to the center of $G$.} 
$G$ may or may not have a nontrivial center and that is not important for us.

Next, let us take a stacked set of links, $\Nu$, which winds around $\Lambda$ in $d-1$ periodic directions 
($\Nu$ is a closed loop when $d=2$ and a closed surface when $d=3$ on the dual lattice, see fig. \ref{spin-vortex}). Hereafter 
such $\Nu$ is called a {\textit{wall}} in distinction from a (center) vortex.
\begin{figure}[hbtp]
\begin{center}
\includegraphics*[width=8cm]{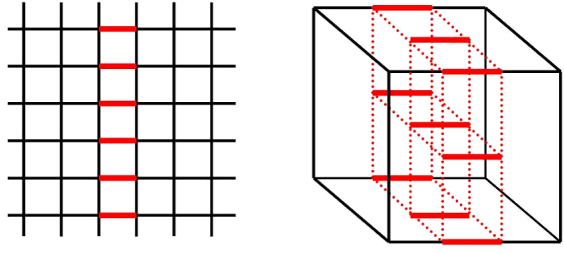}
\end{center}\vspace{-16pt}
\caption{$\Nu$ in two- and three- dimensions.}
\label{spin-vortex}
\end{figure}

The twisted partition function associated to $g\in G$ is given by
\begin{e}\hspace{30pt}
Z_\Lambda(g)[\Nu]\equiv\jint{}{}\prod_{x\in\Lambda}d\phi_x\ \exp\left(
\jsum{y\in\Lambda}{}\jsum{\mu=1}{d}\Bigg[
(1-\delta[\Nu]_{y,\mu})A(\phi_y,\phi_{y+\hat\mu})+\delta[\Nu]_{y,\mu}A(\phi_y,g\phi_{y+\hat\mu})\Bigg]
\right).		\label{twisted_pt}
\end{e}
The symbol $\delta[\Nu]_{y,\mu}$ is defined to be $=1$ 
if the link $\langle y,y+\hat\mu\rangle$ is contained in $\Nu$ and $=0$ otherwise. 
It is not difficult to see that $\Nu$ cannot be removed 
from $\Lambda$ by local redefinition of variables $\phi\to g\phi$. 
Let $\hat G$ denote the set of irreducible unitary representations of $G$. 
The wall creation operators $\{\mathcal O(g)[\Nu]|\,g\in G\}$ and their duals 
$\{\mathcal F_R(g)[\Nu]|\,g\in G,\,R\in\hat G\}$ are defined as follows:
\begin{align}
\langle\mathcal O(g)[\Nu]\rangle&\equiv\jfrac{Z_\Lambda(g)[\Nu]}{Z_\Lambda},
\\
\Longleftrightarrow \ \ \ 
\mathcal O(g)[\Nu]&=\exp\Big(
\jsum{y\in\Lambda}{}\jsum{\mu=1}{d}\delta[\Nu]_{y,\mu}\Big[
-A(\phi_y,\phi_{y+\hat\mu})+A(\phi_y,g\phi_{y+\hat\mu})\Big]\Big),
\\
\mathcal F_R(g)[\Nu]
&\equiv({\rm{dim\,}}R)\jint{G}{}dx\,\mathcal O(gx)[\Nu] \chi_R(x).
\end{align}
\begin{lem}If the walls $\Nu,\,\Nu'$ are homologous
\footnote{Plural walls are called \textit{homologous} if and only if 
they wind around the same periodic directions of $\Lambda$.}, we have
\begin{gather}
\langle\mathcal O(g)[\Nu]\cdot \mathcal O(g')[\Nu']\rangle=\langle\mathcal O(gg')[\Nu]\rangle
=\langle\mathcal O(g'g)[\Nu]\rangle,     \label{aa}
\\
\langle\mathcal F_R(g)[\Nu]\cdot \mathcal F_{R'}(g')[\Nu']\rangle
=\delta_{RR'}\langle\mathcal F_R(gg')[\Nu]\rangle
=\delta_{RR'}\langle\mathcal F_R(g'g)[\Nu]\rangle,         \label{aaa}
\\
\langle\mathcal O(g)[\Nu]\rangle=\jsum{R\in\hat G}{}\langle\mathcal F_R(g)[\Nu]\rangle,   \label{aaaa}
\\
1=\jsum{R\in\hat G}{}\langle\mathcal F_R(\mathbf{1})[\Nu]\rangle.      \label{52*}  
\end{gather}
(\ref{aa}), (\ref{aaa}) imply that $\langle\mathcal O(g)[\Nu]\rangle$ and 
$\langle\mathcal F_R(g)[\Nu]\rangle$ are class functions on $G$.
\end{lem}
\begin{proof}
(\ref{aa}) is trivial, since the relative position of walls can be reversed owing to the periodic 
boundary condition (see fig. \ref{torus}). (\ref{aaa}) can be shown by exploiting 
the invariance of the Haar measure:
\begin{align}
\langle\mathcal F_R(g)\mathcal F_{R'}(g')\rangle
&=({\rm{dim\,}}R)({\rm{dim\,}}R')\left\langle\jint{G}{}dx\ \mathcal O(gx)\chi_R(x)
\jint{G}{}dy\ \mathcal O(g'y)\chi_{R'}(y)\right\rangle
\\
&=({\rm{dim\,}}R)({\rm{dim\,}}R')\jint{G}{}dx
\jint{G}{}dy\ \langle\mathcal O(gxg'y)\rangle\chi_R(x)\chi_{R'}(y)
\\
&=({\rm{dim\,}}R)({\rm{dim\,}}R')\jint{G}{}dx\ \langle\mathcal O(x)\rangle
\jint{G}{}dy\ \chi_R((gg')^{-1}xy^{-1})\chi_{R'}(y)
\\
&=\delta_{RR'}({\rm{dim\,}}R)\jint{G}{}dx\ \langle\mathcal O(x)\rangle\chi_R((gg')^{-1}x)
\\
&=\delta_{RR'}\langle\mathcal F_R(gg')\rangle.
\end{align}
Finally, 
(\ref{aaaa}) is a consequence of the Peter-Weyl theorem \cite{KO} according to which any $f\in L^2(G)$ can be represented as
\begin{e}
f(g)=\jsum{R\in\hat G}{}({\rm{dim\,}}R)\jint{G}{}dx\ f(xg)\chi_R(x).
\end{e}
(\ref{52*}) trivially follows from (\ref{aaaa}).
\end{proof}
\begin{figure}[hbtp]\hspace{20pt}
\begin{center}
\includegraphics*[width=10cm]{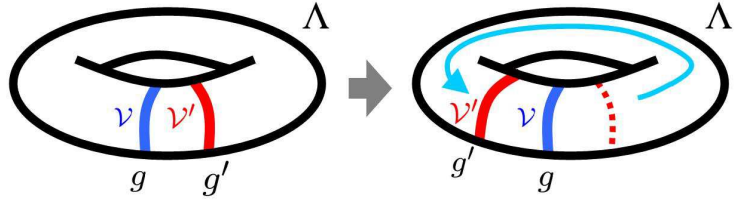}
\end{center}\vspace{-8pt}
\caption{Relative position of $\Nu$ and $\Nu'$ can be reversed by using the periodicity of $\Lambda$.}
\label{torus}
\end{figure}
Note that (\ref{aa}) cannot be proved in general when other operators are inserted, that is,
\begin{e}
\langle\mathcal O(g)[\Nu]\cdot \mathcal O(g')[\Nu']\dots\rangle\not=\langle\mathcal O(gg')[\Nu]\dots\rangle\not
=\langle\mathcal O(g'g)[\Nu]\dots\rangle,
\end{e}
in general, where $\dots$ denote additional insertions. It is because the proof of (\ref{aa}) involves 
a sequence of changes of variables. Thus (\ref{aa}) could be proved if inserted operators are 
invariant under the changes of variables.

Take an arbitrary $(d-1)$-dimensional hyperplane $\pi$ 
defined by $x^\mu=m,\ m\in\mathbb{Z}$ with $m$ and $\mu$ fixed. 
Denote by $\theta$ the reflection about $\pi$. 
\begin{lem}
If the wall $\Nu$ seen on the dual lattice is also perpendicular to the $x^\mu$-axis 
(see the left of fig. \ref{reflection_figure}), we have
\begin{gather}
\theta\Big[\mathcal O(g)[\Nu]\Big]=\mathcal O(g^{-1})[\Nu^\theta],\label{-}
\\
\langle\mathcal O(g)[\Nu]\rangle=\langle\mathcal O(g^{-1})[\Nu]\rangle,\label{--}
\\
\langle\mathcal F_R(g)[\Nu]\rangle
=\overline{\langle\mathcal F_R({{g^{-1}}})[\Nu]\rangle}\ \ {\rm{and}}\ \ 
\langle\mathcal F_R({\mathbf{1}})[\Nu]\rangle\in{\mathbf{R}},\label{---}
\\
\theta\Big[\mathcal F_R(g)[\Nu]\Big]=\mathcal F_R(g^{-1})[\Nu^\theta],\label{----}
\end{gather}
where $\Nu^\theta \equiv \theta[\Nu]$ and $\mathbf{1}$ denotes the unit element of $G$. 
\end{lem}
\begin{figure}[hbtp]\hspace{20pt}
\begin{center}
\includegraphics*[width=11cm]{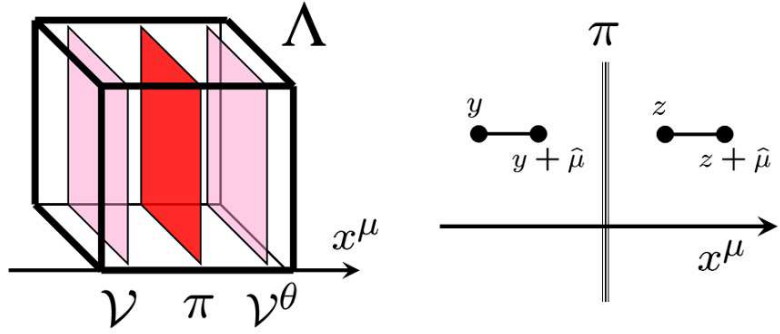}
\end{center}\vspace{-16pt}
\caption{Insertion of a reflection plane $\pi$.}
\label{reflection_figure}
\end{figure}
\begin{proof}
From the right of fig. \ref{reflection_figure} we observe
\begin{align}
\theta\Big[A(\phi_y,g\phi_{y+\hat\mu})\Big]&=A(\phi_{z+\hat\mu},g\phi_{z})
\\
&=A(\phi_z,g^{-1}\phi_{z+\hat\mu}).
\end{align}
The twist $g$ changed to $g^{-1}$, hence (\ref{-}) is proved. (\ref{--}) immediately follows from (\ref{-}).

Using (\ref{--}) we can show
\begin{align}
\jint{G}{}dx\ \langle\mathcal O(gx)[\Nu]\rangle\chi_R(x)
&=\jint{G}{}dx\ \langle\mathcal O((gx)^{-1})[\Nu]\rangle\chi_R(x)
\\
&=\jint{G}{}dx\ \langle\mathcal O(x^{-1}g^{-1})[\Nu]\rangle\chi_R(x)
\\
&=\jint{G}{}dx\ \langle\mathcal O(xg^{-1})[\Nu]\rangle\chi_R(x^{-1})
\\
&=\jint{G}{}dx\ \langle\mathcal O(g^{-1}x)[\Nu]\rangle\overline{\chi_R(x)},
\end{align}
therefore (\ref{---}) is proved. In the last step we used the fact that $R$ is a unitary representation. 
(\ref{----}) is obvious from (\ref{---}).
\end{proof}
Though it seems hard to find more properties on general grounds, further nontrivial result can be obtained 
if we exploit the site-reflection positivity of the measure of (\ref{partition_}).
\begin{lem}
\begin{gather}
\langle\mathcal O(g)[\Nu]\rangle,\ \,\langle\mathcal F_R(\mathbf{1})[\Nu]\rangle\ \in [0,1],  \label{ai}
\\
|\langle\mathcal F_R({{g}})[\Nu]\rangle|\leq \langle\mathcal F_R({\mathbf{1}})[\Nu]\rangle.   \label{ai'}
\end{gather}
\end{lem}
\begin{proof}
Let $\Nu^\theta\equiv\theta[\Nu]$. With the aid of 
(\ref{aa}), (\ref{-}) and the Schwarz inequality $|\langle F\rangle|\leq\langle F\theta F\rangle^{1/2}$ we get
\begin{align}
\langle\mathcal O(g)[\Nu]\rangle&\leq\langle\mathcal O(g)[\Nu]\theta\Big[O(g)[\Nu]\Big]\rangle^{1/2}
\\
&=\langle\mathcal O(g)[\Nu]O(g^{-1})[\Nu^\theta]\rangle^{1/2}=1.
\end{align}
Next, using (\ref{aaa}) and (\ref{----}) yields
\begin{align}
\langle\mathcal F_R(\mathbf{1})[\Nu]\rangle
&=\langle\mathcal F_R(\mathbf{1})[\Nu]\cdot \mathcal F_R(\mathbf{1})[\Nu^\theta]\rangle
\\
&=\langle\mathcal F_R(\mathbf{1})[\Nu]\cdot\theta\Big[\mathcal F_R(\mathbf{1})[\Nu]\Big]\rangle\geq 0.\label{ue}
\end{align}
(\ref{ue}) combined with (\ref{52*}) yields (\ref{ai}).

Finally, to show (\ref{ai'}) we use (\ref{----}) and the Schwarz inequality 
$|\langle F\theta G\rangle|\leq\langle F\theta F\rangle^{1/2}
\langle G\theta G\rangle^{1/2}$ as follows: letting $\Nu'$ denote a wall homologous to $\Nu$, we have
\begin{align}
|\langle\mathcal F_R(g)[\Nu]\rangle|
&=|\langle\mathcal F_R(g)[\Nu]\mathcal F_R({\mathbf{1}})[\Nu']\rangle|
\\
&\leq \langle\mathcal F_R(g)[\Nu]\mathcal F_R(g^{-1})[\Nu^\theta]\rangle^{1/2}
\langle\mathcal F_R({\mathbf{1}})[\Nu']\mathcal F_R({\mathbf{1}})[\Nu'^\theta]\rangle^{1/2}
\\
&=\langle\mathcal F_R({\mathbf{1}})[\Nu]\rangle^{1/2}\langle\mathcal F_R({\mathbf{1}})[\Nu']\rangle^{1/2}
\\
&=\langle\mathcal F_R({\mathbf{1}})[\Nu]\rangle.
\end{align}
\end{proof}

\subsection{TY inequality in spin systems}

A natural counterpart in spin systems of Wilson loops in LGT 
is a two-point correlation function $\Gamma(\phi_x,\phi_y)$ as explained in the Introduction. 
$\Gamma$ will decay exponentially with a mass gap (in symmetric phase) 
while decay algebraically without a mass gap (in a spontaneous symmetry breaking phase or Kosterlitz-Thouless-type phase). 
If $\langle\mathcal O(g)\rangle$ defined above converges to 1 in the thermodynamic limit, it follows that arbitrarily 
huge domain walls grow with little cost and eventually drive the system to the disordered phase with a mass gap 
(see fig. \ref{Bloch}). $\langle\mathcal O(g)\rangle\to 1\ \,(|\Lambda|\to\infty)$ can also be regarded 
as a sign of insensitivity of the system to boundary conditions, which indicates the absence of 
massless particles.

If (as in Ising-like models) we assume that the intersection of a correlation line with a wall 
changes the sign of $\Gamma$, a small closed wall gives no contribution ($(-1)^2=1$) while 
a huge wall can give $(-1)$. If each link on the correlation line of total length $L$ 
is assumed to intersect with a wall independently with probability $p$, we obtain
\begin{e}
\langle\Gamma\rangle\sim 
\jsum{k=0}{L}\begin{pmatrix}L\\k\end{pmatrix}(-1)^kp^k(1-p)^{L-k}=(1-2p)^{L}\sim e^{-mL},\hspace{30pt}
m=-\log(1-2p).
\end{e}
\begin{figure}[hbtp]\hspace{20pt}\vspace{-15pt}
\begin{center}
\includegraphics*[width=5cm]{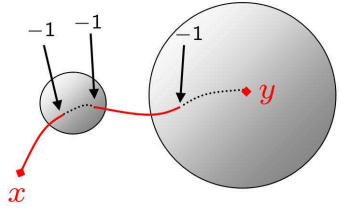}
\end{center}\vspace{-14pt}
\caption{Disorder being caused by walls.}
\label{Bloch}
\end{figure}

Our goal in this section is to elevate the above heuristic relation of disorder and a mass gap 
to a mathematically rigorous inequality.

Suppose that the explicit form of the correlation function $\Gamma$ is given by
\begin{e}
\Gamma_\mu(x;n)=\jsum{\alpha}{}f_\alpha(\phi_x)\overline{f_\alpha(\phi_{x+n\hat\mu})},\hspace{30pt}n\in\mathbf{N},
\end{e}
where $f$ is an arbitrary function from the order-parameter space to $\mathbf{C}$ 
endowed with generic indices $\{\alpha\}$. Here $f_\alpha$ is meant to specify, for example, 
the $\alpha$-th component of a vector spin in $O(N)$-like models, or the $\alpha$-th matrix element of a 
matrix spin in PCM-like models. 
Note that 
$\langle\Gamma_\mu(x;n)\rangle$ is independent of $x$ due to the translational invariance of the system. 
An important requirement on $\Gamma_\mu$ is its invariance under $G$
:
\begin{e}
\jsum{\alpha}{}f_\alpha(\phi_x)\overline{f_\alpha(\phi_{x+n\hat\mu})}
=\jsum{\alpha}{}f_\alpha(g\phi_x)\overline{f_\alpha(g\phi_{x+n\hat\mu})},\hspace{20pt}g\in G.   \label{G-i}
\end{e}
Another requirement, which is \textit{truly indispensable for all the following development}, is
\begin{e}
\jint{G}{}dg\,f_\alpha(g\phi)=0\ \ \ \ \ {\rm{for}}\ \ ^\forall\phi.        \label{technical}
\end{e}
If (\ref{technical}) were not satisfied, we should replace $f_\alpha(\phi)$ by $f'_\alpha(\phi)\equiv 
f_\alpha(\phi)-\jint{G}{}dg\,f_\alpha(g\phi)$.
\begin{thm}\label{thm}
Assume (\ref{G-i}), (\ref{technical}) and the existence of $k\in\mathbf{N}$ such that $L_\mu=2^kn$, 
with $L_\mu$ the extent of $\Lambda$ in the $x^\mu$-direction.\footnote{
It does not seem to be very restrictive; in general we believe in the existence of the limit $L_\mu\to\infty$ 
for physical observables independent of the way we take $L_\mu\to\infty$.} 
Then we have
\begin{e}
\jfrac{|\langle\Gamma_\mu(x;n)\rangle|}{\langle\Gamma_\mu(x;0)\rangle}
\leq 2
\left\{\jfrac{\langle\Gamma_\mu(x;0)^2\rangle}{\langle \Gamma_\mu(x;0)\rangle^2}\right\}^{n/L_\mu}
\left\{1-\jint{G}{}dg\,\langle\mathcal O(g)[\Nu]\rangle\right\}^{n/L_\mu},
\label{problem}
\end{e}
where $\Nu$ is a wall perpendicular to the $x^\mu$-direction
.\footnote{$0<\langle\Gamma_\mu(n)\rangle$ can be proved by the site-reflection positivity if $n$ is even.}
\end{thm}
\begin{proof}
We first decompose the correlation function into two parts:
\begin{e}
\langle\Gamma_\mu(x;n)\rangle=\langle\Gamma_\mu(x;n)(1-\mathcal F_T(\mathbf{1})[\Nu])\rangle+
\langle\Gamma_\mu(x;n)\mathcal F_T(\mathbf{1})[\Nu]\rangle,   \label{-84-}
\end{e}
where $T$ denotes the trivial representation, i.e. $\mathcal F_T(\mathbf{1})[\Nu]=\jint{G}{}dg\,\mathcal O(g)[\Nu]$. 
Our basic idea here is to apply the Schwarz inequality
\begin{e}
|\langle F\theta G\rangle|\leq \langle F\theta F\rangle^{1/2}\langle G\theta G\rangle^{1/2}   \label{SI}
\end{e}
to each term of (\ref{-84-}). 
The procedure afterward is represented graphically in fig. \ref{RPRP} 
where $(1-\mathcal F_T(\mathbf{1})[\Nu])$ is indicated by red segments and a blue line is drawn to guide the eye. 
Let $y\equiv x+n\hat\mu$ and $z\equiv x+2n\hat\mu$.
Consider a reflection $\theta$ about $\pi$ (a hyperplane perpendicular to $\hat\mu$ and lying at $y$ ). 
Using (\ref{aaa}), (\ref{----}), (\ref{G-i}) and (\ref{SI}) we get 
\begin{align}
&\jfrac{|\langle\Gamma_\mu(x;n)(1-\mathcal{F}_T(\mathbf{1})[\Nu])\rangle|}{\langle\Gamma_\mu(x;0)\rangle}
\\
=\ &\jfrac{1}{\langle \Gamma_\mu(x;0)\rangle}\Big|\jsum{\alpha}{}\langle f_\alpha(\phi_x)\overline{f_\alpha(\phi_{y})}
(1-\mathcal{F}_T(\mathbf{1})[\Nu])\rangle\Big|
\\
\leq\ &\jfrac{1}{\langle \Gamma_\mu(x;0)\rangle}\jsum{\alpha}{}\Big
\langle f_\alpha(\phi_x)(1-\mathcal{F}_T(\mathbf{1})[\Nu])\cdot 
\theta\Big[ f_\alpha(\phi_x)(1-\mathcal{F}_T(\mathbf{1})[\Nu])\Big]\Big\rangle^{1/2}
\Big\langle \overline{f_\alpha(\phi_{y})}
\theta\Big[\overline{f_\alpha(\phi_{y})}\Big]\Big\rangle^{1/2}
\\
=\ &\jfrac{1}{\langle \Gamma_\mu(x;0)\rangle}\jsum{\alpha}{}
\big\langle f_\alpha(\phi_x)(1-\mathcal{F}_T(\mathbf{1})[\Nu])
\overline{f_\alpha(\phi_z)}(1-\mathcal{F}_T(\mathbf{1})[\Nu^\theta])\big\rangle^{1/2}
\big\langle \overline{f_\alpha(\phi_{y})}f_\alpha(\phi_{y})\big\rangle^{1/2}
\\
\leq\ &\jfrac{1}{\langle \Gamma_\mu(x;0)\rangle}
\Big\langle \jsum{\alpha}{}f_\alpha(\phi_x)\overline{f_\alpha(\phi_z)}(1-\mathcal{F}_T(\mathbf{1})[\Nu])\Big\rangle^{1/2}
\Big\langle \jsum{\alpha}{}\overline{f_\alpha(\phi_{y})}f_\alpha(\phi_{y})\Big\rangle^{1/2}
\\
=\ &\left\{\jfrac{\langle\Gamma_\mu(x;2n)(1-\mathcal{F}_T(\mathbf{1})[\Nu])\rangle}
{\langle \Gamma_\mu(x;0)\rangle}\right\}^{1/2}.
\end{align}
\begin{figure}[hbtp]
\begin{center}
\includegraphics*[width=10cm]{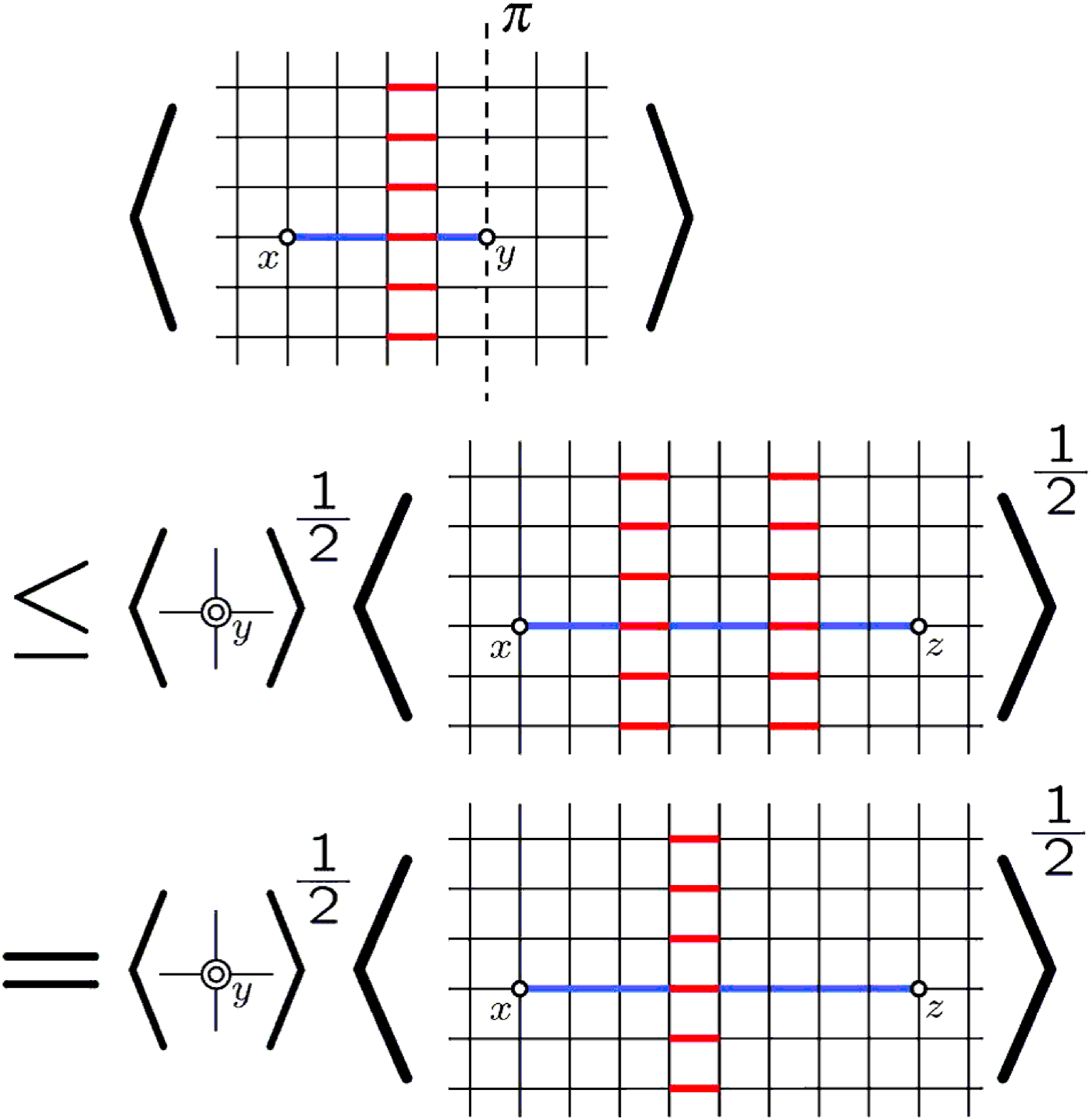}
\end{center}\vspace{-15pt}
\caption{Schwarz inequality enables us to double the distance of points in the correlation function.}
\label{RPRP}
\end{figure}

Iterating this procedure for $(k-1)$ times yields (note $L_\mu=2^kn$)
\begin{e}
\jfrac{|\langle\Gamma_\mu(n)(1-\mathcal{F}_T(\mathbf{1})[\Nu])\rangle|}{\langle\Gamma_\mu(0)\rangle}
\leq \left\{\jfrac{\langle\Gamma_\mu(L_\mu/2)(1-\mathcal{F}_T(\mathbf{1})[\Nu])\rangle}
{\langle \Gamma_\mu(0)\rangle}\right\}^{2n/L_\mu}.
\end{e}
\begin{figure}[hbtp]
\begin{center}
\includegraphics*[width=9.0cm]{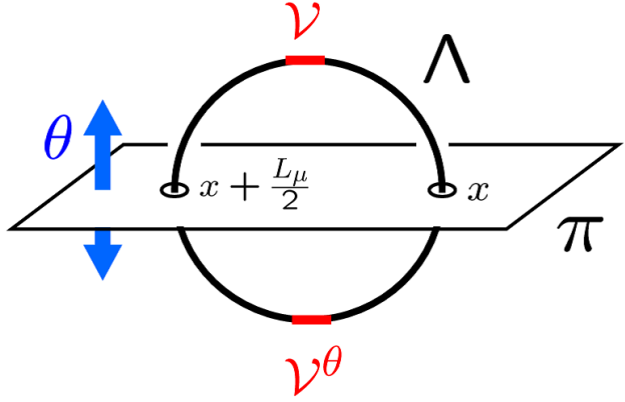}
\end{center}\vspace{-15pt}
\caption{$\theta$ is a reflection about $\pi$ that bisects $\Lambda$. Here $\Lambda$ is shown as a 
one-dimensional chain for simplicity.}
\label{1D_reflection}
\end{figure}

Finally, let us define the reflection $\theta$ w.r.t. the hyperplane $\pi$ which runs through $x$ and $x+(L_\mu/2)\hat\mu$ 
(thus $\theta[x]=x$ and $\theta[x+(L_\mu/2)\hat\mu]=x+(L_\mu/2)\hat\mu$, see fig.\ref{1D_reflection}). Then we obtain
\begin{align}
&\jfrac{\langle\Gamma_\mu(x;L_\mu/2)(1-\mathcal{F}_T(\mathbf{1})[\Nu])\rangle}
{\langle \Gamma_\mu(x;0)\rangle}
\\
\leq\ &\jfrac{
\langle \Gamma_\mu(x;L_\mu/2)\theta[\Gamma_\mu(x;L_\mu/2)]\rangle^{1/2}
\langle(1-\mathcal{F}_T(\mathbf{1})[\Nu])(1-\mathcal{F}_T(\mathbf{1})[\Nu^\theta])\rangle^{1/2}
}{\langle\Gamma_\mu(x;0)\rangle}
\\
\leq\ &\jfrac{\langle \Gamma_\mu(x;0)^2\rangle^{1/2}
\left\{1-\langle \mathcal{F}_T(\mathbf{1})[\Nu]\rangle\right\}^{1/2}
}{\langle\Gamma_\mu(x;0)\rangle}
\left(=\left\{\jfrac{\langle\Gamma_\mu(x;0)^2\rangle}{\langle \Gamma_\mu(x;0)\rangle^2}\right\}^{1/2}
\left\{1-\langle \mathcal{F}_T(\mathbf{1})[\Nu]\rangle\right\}^{1/2}\right),    \label{shie}
\end{align}
thus
\begin{e}
0\leq\jfrac{\langle\Gamma_\mu(x;n)(1-\mathcal{F}_T(\mathbf{1})[\Nu])\rangle}{\langle\Gamma_\mu(x;0)\rangle}
\leq \left\{\jfrac{\langle\Gamma_\mu(x;0)^2\rangle}{\langle \Gamma_\mu(x;0)\rangle^2}\right\}^{n/L_\mu}
\left\{1-\langle \mathcal{F}_T(\mathbf{1})[\Nu]\rangle\right\}^{n/L_\mu}.
\label{rei}
\end{e}
In deriving (\ref{shie}) we used $\langle \Gamma_\mu(x;L_\mu/2)\theta[\Gamma_\mu(x;L_\mu/2)]\rangle
\leq \langle \Gamma_\mu(x;0)^2\rangle$. This can be shown as follows:
\begin{align}
&\langle \Gamma_\mu(x;L_\mu/2)\theta[\Gamma_\mu(x;L_\mu/2)]\rangle
\\
=\ &\jsum{\alpha,\beta}{}
\Big\langle f_\alpha(\phi_x)\overline{f_\alpha(\phi_{x+(L_\mu/2)\hat\mu})}
\cdot \theta\Big[f_\beta(\phi_x)\overline{f_\beta(\phi_{x+(L_\mu/2)\hat\mu})}\Big]\Big\rangle
\\
=\ &\jsum{\alpha,\beta}{}\big\langle f_\alpha(\phi_x)\overline{f_\alpha(\phi_{x+(L_\mu/2)\hat\mu})}\ 
\overline{f_\beta(\phi_x)}f_\beta(\phi_{x+(L_\mu/2)\hat\mu})\big\rangle
\\
\leq\ &\jsum{\alpha,\beta}{}\Big\langle f_\alpha(\phi_x)\overline{f_\beta(\phi_x)}
\cdot\theta\Big[f_\alpha(\phi_x)\overline{f_\beta(\phi_x)}\Big]\Big\rangle^{1/2}
\Big\langle\overline{f_\alpha(\phi_{x+(L_\mu/2)\hat\mu})}f_\beta(\phi_{x+(L_\mu/2)\hat\mu})
\cdot\theta\Big[\overline{f_\alpha(\phi_{x+(L_\mu/2)\hat\mu})}f_\beta(\phi_{x+(L_\mu/2)\hat\mu})\Big]\Big\rangle^{1/2}
\\
=\ &\jsum{\alpha,\beta}{}\langle|f_\alpha(\phi_x)|^2|f_\beta(\phi_x)|^2\rangle^{1/2}
\langle|f_\alpha(\phi_{x+(L_\mu/2)\hat\mu})|^2|f_\beta(\phi_{x+(L_\mu/2)\hat\mu})|^2\rangle^{1/2}
\\
\leq\ &\Big\langle\jsum{\alpha,\beta}{}|f_\alpha(\phi_x)|^2|f_\beta(\phi_x)|^2\Big\rangle^{1/2}
\Big\langle\jsum{\alpha,\beta}{}|f_\alpha(\phi_{x+(L_\mu/2)\hat\mu})|^2|f_\beta(\phi_{x+(L_\mu/2)\hat\mu})|^2
\Big\rangle^{1/2}
\\
=\ &\langle\Gamma_\mu(x;0)^2\rangle^{1/2}\langle\Gamma_\mu(x+(L_\mu/2);0)^2\rangle^{1/2}=\langle\Gamma_\mu(x;0)^2\rangle.
\end{align}

Next we have to estimate 
$\jfrac{\langle\Gamma_\mu(x;n)\mathcal{F}_T(\mathbf{1})[\Nu]\rangle}{\langle\Gamma_\mu(x;0)\rangle}$. 
The outline is similar to the previous case, but additional intricacies occur.
\begin{align}
0\leq\ &\langle\Gamma_\mu(x;n)\mathcal{F}_T(\mathbf{1})[\Nu]\rangle
\\
=\ &\jsum{\alpha}{}\jint{G}{}dg\,\langle f_\alpha(\phi_x)
\overline{f_\alpha(\phi_{y})}\mathcal{O}(g)[\Nu]\rangle
\intertext{Let us consider another wall $\Nu'$ that is parallel to $\Nu$ and bookends $x$ with $\Nu$ 
(see fig.\ref{move_nu}). Moving $\Nu$ to $\Nu'$ passing over $\phi_x$ and using (\ref{technical}), we get}
=\ &\jsum{\alpha}{}\jint{G}{}dg\,\langle f_\alpha(g\phi_x)
\overline{f_\alpha(\phi_{y})}\mathcal{O}(g)[\Nu']\rangle
\label{hingis-0}
\\
=\ &\jsum{\alpha}{}\jint{G}{}dg\,\Big\langle f_\alpha(g\phi_x)
\overline{f_\alpha(\phi_{y})}\Big(\mathcal{O}(g)[\Nu']-\mathcal{F}_T(\mathbf{1})[\Nu']\Big)\Big\rangle.
\label{hingis}
\end{align}
This step (from (\ref{hingis-0}) to (\ref{hingis})) is the most nontrivial operation in this proof.
\begin{figure}[hbtp]
\begin{center}
\includegraphics*[width=6.0cm]{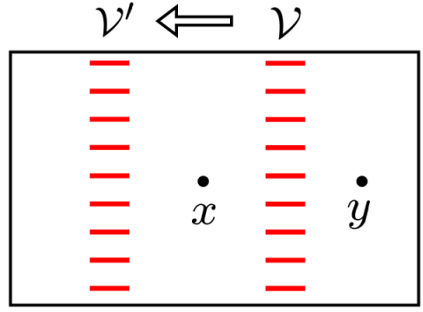}
\end{center}\vspace{-15pt}
\caption{By redefinition of variables, $\Nu$ is moved to $\Nu'$, passing over the site $x$.}
\label{move_nu}
\end{figure}
Using (\ref{SI}) w.r.t. the hyperplane which runs through $y$, with $z=\theta[x]$, yields
\begin{align}
=\ &\jsum{\alpha}{}\jint{G}{}dg\,\Big\langle f_\alpha(g\phi_x)
\Big(\mathcal{O}(g)[\Nu']-\mathcal{F}_T(\mathbf{1})[\Nu']\Big)
\cdot \theta\Big[f_\alpha(g\phi_x)
\Big(\mathcal{O}(g)[\Nu']-\mathcal{F}_T(\mathbf{1})[\Nu']\Big)\Big]
\Big\rangle^{1/2}\big\langle \overline{f_\alpha(\phi_{y})}f_\alpha(\phi_{y})\big\rangle^{1/2}
\\
=\ &\jsum{\alpha}{}\jint{G}{}dg\,\Big\langle f_\alpha(g\phi_x)
\Big(\mathcal{O}(g)[\Nu']-\mathcal{F}_T(\mathbf{1})[\Nu']\Big)
\overline{f_\alpha(g\phi_z)}\Big(\mathcal{O}(g^{-1})[{\Nu'}^{\theta}]-\mathcal{F}_T(\mathbf{1})[{\Nu'}^\theta]\Big)
\Big\rangle^{1/2}\big\langle \overline{f_\alpha(\phi_{y})}f_\alpha(\phi_{y})\big\rangle^{1/2}
\\
=\ &\jsum{\alpha}{}\jint{G}{}dg\,\Big\langle f_\alpha(g\phi_x)\overline{f_\alpha(g\phi_z)}
\Big(1-\mathcal O(g)[\Nu']\mathcal{F}_T(\mathbf{1})[{\Nu'}^\theta]
-\mathcal O(g^{-1})[{\Nu'}^\theta]\mathcal{F}_T(\mathbf{1})[\Nu']+\mathcal{F}_T(\mathbf{1})[\Nu']\Big)
\Big\rangle^{1/2}\notag
\\
&\hspace{30pt}\times\big\langle \overline{f_\alpha(\phi_{y})}f_\alpha(\phi_{y})\big\rangle^{1/2}
\\
\leq \ &\langle \Gamma_\mu(y;0)\rangle^{1/2}\jint{G}{}dg\,\Big\langle \Gamma_\mu(x;2n)
\Big(1-\mathcal O(g)[\Nu']\mathcal{F}_T(\mathbf{1})[{\Nu'}^\theta]
-\mathcal O(g^{-1})[{\Nu'}^\theta]\mathcal{F}_T(\mathbf{1})[\Nu']+\mathcal{F}_T(\mathbf{1})[\Nu']\Big)
\Big\rangle^{1/2}
\\
\leq \ &\langle \Gamma_\mu(x;0)\rangle^{1/2}\Big[\jint{G}{}dg\,
\Big\langle \Gamma_\mu(x;2n)
\Big(1-\mathcal O(g)[\Nu']\mathcal{F}_T(\mathbf{1})[{\Nu'}^\theta]
-\mathcal O(g^{-1})[{\Nu'}^\theta]\mathcal{F}_T(\mathbf{1})[\Nu']+\mathcal{F}_T(\mathbf{1})[\Nu']\Big)
\Big\rangle\Big]^{1/2}
\\
=\ &\langle \Gamma_\mu(x;0)\rangle^{1/2}\Big[\langle \Gamma_\mu(x;2n)\rangle
-\langle\Gamma_\mu(x;2n)\mathcal{F}_T(\mathbf{1})[\Nu']\mathcal{F}_T(\mathbf{1})[{\Nu'}^\theta]\rangle
-\langle\Gamma_\mu(x;2n)\mathcal F_T(\mathbf{1})[{\Nu'}^\theta]\mathcal{F}_T(\mathbf{1})[\Nu']\rangle
\notag
\\
&\hspace{50pt}
+\langle\Gamma_\mu(x;2n)\mathcal{F}_T(\mathbf{1})[\Nu']\rangle\Big]^{1/2}
\\
=\ &\langle \Gamma_\mu(x;0)\rangle^{1/2}\langle \Gamma_\mu(x;2n)(1-\mathcal{F}_T(\mathbf{1})[\Nu'])\rangle^{1/2},
\end{align}
\begin{e}\therefore \ \ 
\jfrac{\langle\Gamma_\mu(x;n)\mathcal{F}_T(\mathbf{1})[\Nu]\rangle}{\langle \Gamma_\mu(x;0)\rangle}
\leq
\left\{\jfrac{\langle \Gamma_\mu(x;2n)(1-\mathcal{F}_T(\mathbf{1})[\Nu'])\rangle}{\langle \Gamma_\mu(x;0)\rangle}\right\}^{1/2}.
\end{e}
Combining this result with (\ref{rei}) yields
\begin{e}
\jfrac{\langle\Gamma_\mu(x;n)\mathcal{F}_T(\mathbf{1})[\Nu]\rangle}{\langle \Gamma_\mu(x;0)\rangle}
\leq
\left\{\jfrac{\langle\Gamma_\mu(x;0)^2\rangle}{\langle \Gamma_\mu(x;0)\rangle^2}\right\}^{n/L_\mu}
\left\{1-\langle \mathcal{F}_T(\mathbf{1})[\Nu']\rangle\right\}^{n/L_\mu}.
\label{asuka}
\end{e}
(\ref{rei}) and (\ref{asuka}) lead to (\ref{problem}).
\end{proof}
Here are a few comments:
\begin{itemize}
\item 
$\langle\mathcal O(g)\rangle$ is a class function on $G$. It implies that 
when $G$ is a compact connected Lie group we can rewrite the integral in the r.h.s. of (\ref{problem}) 
as an integral over the maximal torus of $G$ with a proper Jacobian, using the 
Weyl's integration formula \cite{KO}. 
\item 
According to (\ref{problem}), \textit{the exponential decay of the wall free energy, i.e.}
\begin{e}\hspace{40pt}
\langle \mathcal O(g)[\Nu]\rangle=1-O(\ee^{-\rho(g) L_\mu}) {\hspace{20pt}\rm{with}\ \ }
\rho(g)>0,\ \ \ ^\forall g\in G, 
\end{e}
\textit{is a sufficient condition for the existence of a mass gap}. Indeed if we define $\overline{\rho}[G]\equiv 
\displaystyle\min_{g\in G}\,\rho(g)$ and take the limit $L_\mu\to\infty$, the r.h.s. of (\ref{problem}) 
converges to $O(1)\times \ee^{-\overline{\rho}[G] n}$, thus $\overline{\rho}[G]$ 
gives a lower bound for the mass gap. It is 
also obvious that $\overline{\rho}[G]\leq \overline{\rho}[H]$ follows if $H\subset G$. 
$\langle \mathcal O(g)[\Nu]\rangle$ would be calculable by, for example, 
Monte Carlo Simulations and weak-coupling-, strong-coupling-, $1/d$- and $1/N$-expansions. 

\item 
We proved the theorem on a square lattice, but it can be easily generalized to other lattices such as a triangular 
lattice.

\item 
Note that theorem \ref{thm} is derived with no knowledge of the interaction except for the reflection positivity. 
The strength of interaction can be made anisotropic, since it respects reflection positivity. 
And it is also correct in arbitrary dimensions.
\end{itemize}

\subsection{Examples}
There are many classes of lattice systems to which our theorem is applicable. One is the class of 
\textit{coset models}, where the field takes values in $G/H$ with $G$ an arbitrary Lie group and $H$ its closed subgroup. 
We present three explicit examples.
\begin{exa}
The $O(N)$ Heisenberg model.
\end{exa}
The partition function is given by
\begin{e}\hspace{60pt}
Z=\jint{}{}\prod_{y\in\Lambda}d\vec\phi_y\ 
\exp\Big(\beta\jsum{x,\mu}{}\vec\phi_x\cdot\vec\phi_{x+\hat\mu}\Big),\hspace{20pt}\vec \phi\in S^{N-1}
=O(N)/O(N-1).
\end{e}
The most natural definition of a correlation function is
\begin{e}
\Gamma_\mu(x;n)=\vec \phi_x\cdot \vec\phi_{x+n\hat\mu}.
\end{e}
In this case $Z_2$ is obviously the smallest possible invariance group; of course, $O(N)$ itself can also be chosen. 
One can check that the requirements (\ref{G-i}), (\ref{technical}) are met for both of them. 
If one is going to choose other arbitrary subgroup of $O(N)$, the correlation function should be redefined properly 
as explained below (\ref{technical}). 
Since the length of spin is normalized, the prefactor in (\ref{problem}) becomes 1.
\begin{exa}
The $CP^{N-1}$ model.
\end{exa}
The partition function is given by\footnote{See also ref.\cite{BPRW}.}
\begin{e}
Z=\jint{}{}\prod_{y\in\Lambda}dP_y\ \exp\Big(\beta\jsum{x,\mu}{}\mathrm{Tr}\,[P_xP_{x+\hat\mu}]\Big),\hspace{20pt}P\in 
CP^{N-1}=U(N)/(U(1)\times U(N-1)).
\end{e}
$P$ is an $N\times N$ matrix and obeys $P_x^2=P_x,\ P_x^\dagger=P_x,\ \mathrm{Tr}\,P_x=1$. (An alternative way is to express 
$P_{ij}=z_i^*z_j$, where $\vec{z}\in \mathbb{C}^N$ and $|\vec z|^2=1$.) The action is invariant under the global 
transformation $P_x\to \Omega P_x\Omega^\dagger$ with $\Omega\in SU(N)$.\footnote{The \textit{true} symmetry 
is $SU(N)/Z_N$, since $\Omega\in Z_N$ does not change $P_x$ at all.} In two dimension, the model is believed to possess a 
nonperturbatively generated mass gap for which, however, no rigorous result is available. 
Defining an appropriate correlation function needs some care; we define
\begin{e}
\Gamma_\mu(x;n)=\mathrm{Tr\,}\left\{[P_x-\jfrac{\mathbf{1}_{N}}{N}][P_{x+n\hat\mu}-\jfrac{\mathbf{1}_{N}}{N}]\right\}
=\mathrm{Tr\,}[P_xP_{x+n\hat\mu}]-\jfrac{1}{N},
\end{e}
where $\mathbf{1}_N$ denotes the unit matrix of size $N\times N$. Then it is easy to confirm (\ref{technical}) for $G=SU(N)$:
\begin{align}
\jint{SU(N)}{}d\Omega\,\Big[\Omega P\Omega^\dagger-\jfrac{\mathbf{1}_N}{N}\Big]_{ij}
&=P_{kl}\jint{SU(N)}{}d\Omega\ \Omega_{ik}{\Omega^\dagger}_{lj}-\jfrac{\delta_{ij}}{N}
\\
&=P_{kl}\jfrac{\delta_{ij}\delta_{kl}}{N}-\jfrac{\delta_{ij}}{N}=0.
\end{align}
%
%
%
\begin{exa}
$G\times G$ principal chiral model (PCM).
\end{exa}
Here $G$ is an arbitrary compact group. The partition function is given by
\begin{e}
Z=\jint{}{}\prod_{y\in\Lambda}dU_y\ \exp\left(\beta
\jsum{x,\mu\subset \Lambda}{}
{\mathrm{Re\,Tr\,}}[U_x(U_{x+\hat\mu})^{-1}]\right),\hspace{20pt}U_x\in G.
\end{e}
This model is invariant under a global $G\times G$ transformation $U\to g_LUg_R^{-1}$. 
The corresponding twisted partition function is given, in agreement with the original definition (\ref{twisted_pt}), 
by ($\Nu$\,:\,wall)
\begin{e}
Z(g)=\jint{}{}\prod_{y\in\Lambda}dU_y\ \exp\left(\beta
\jsum{x,\mu\subset \Lambda\setminus\Nu}{}{\mathrm{Re\,Tr\,}}[U_x(U_{x+\hat\mu})^{-1}]
+\beta
\jsum{x,\mu\subset \Nu}{}{\mathrm{Re\,Tr\,}}[U_x(gU_{x+\hat\mu})^{-1}]\right),\hspace{20pt}g\in G.
\end{e}
The correlation function is defined as
\begin{e}
\Gamma_\mu(x;n)=\chi_R(U_xU^{-1}_{x+n\hat\mu}),
\end{e}
where $R$ denotes an irreducible unitary representation of $G$. If we choose $G$ itself, then (\ref{technical}) is satisfied 
if and only if $R$ is a nontrivial representation. This is a direct consequence of the Schur orthogonality relation in 
representation theory; see Corollary 4.10 in ref.\cite{KO}. 
If $G=SU(N)$ and we choose $Z_N$, then (\ref{technical}) is satisfied iff 
$R$ has a nonzero $N$-ality. Anyway the prefactor in (\ref{problem}) becomes 1. 
We emphasize that $G$ need not have a nontrivial center.

As an example other than coset models, we only present
\begin{exa}
$SU(N_f)\times SU(N_f)$ linear sigma model.
\end{exa}
Usually this model is used to study the chiral phase transition \cite{PW}. 
On the lattice, the partition function is given by
\begin{gather}
\hspace{-10pt}
Z=\jint{}{}\prod_{y\in\Lambda}d\Phi_y\ \exp\Big(
\beta_1\jsum{x,\mu}{}
{\mathrm{Re\,Tr\,}}(\Phi_x\Phi^\dagger_{x+\hat\mu})-\beta_2\jsum{x\in\Lambda}{}{\mathrm{Tr\,}}(\Phi_x^\dagger\Phi_x)
-\lambda_1\jsum{x\in\Lambda}{}\big[{\mathrm{Tr\,}}(\Phi_x^\dagger\Phi_x)\big]^2
-\lambda_2\jsum{x\in\Lambda}{}{\mathrm{Tr\,}}(\Phi_x^\dagger\Phi_x)^2     \Big),\\
\hspace{20pt}\Phi_x\in {\mathrm{M}}(N_f,\mathbb{C}).
\end{gather}
This model is invariant under $\Phi\to \ee^{i\alpha}g_L\Phi g^\dagger_R$ with 
$\ee^{i\alpha}\in U_A(1)$ and $g_L,g_R\in SU(N_f)$.

It might be the case that \textit{whether the upper bound (\ref{problem}) gives an exponential decay or not 
depends on the choice of $G$}, a point worth further study.
\vspace{5pt}
\\
\hfil\hspace{80pt}*\hspace{80pt}*\hspace{80pt}*\hspace{80pt}*\hspace{80pt}

Although we proved theorem \ref{thm} only in the case of the nearest-neighbor interaction, 
\textbf{we can prove it even in the presence of non-nearest-neighbor and multi-site interactions if some appropriate 
conditions are satisfied.} To make the argument concrete, let us consider the $SU(N)\times SU(N)$ PCM and suppose 
that $Z_N\subset SU(N)$ was chosen as a symmetry group for theorem \ref{thm}. 
There are two crucial conditions: one is site-reflection positivity and the other is 
that twists and their algebra (see \ref{Basic formulation}) should remain well-defined. 
Here is a partial list of possible extensions (on a square lattice):
\begin{enumerate}
\item The multi-site interaction term between four variables on the same plaquette,
\begin{e}
\mathrm{Re\,}\chi_R(U_xU^\dagger_{x+\hat\mu} U_{x+\hat\mu+\hat\nu}U^\dagger_{x+\hat\nu}),
\end{e}
can be added to the action without spoiling theorem \ref{thm} if the $N$-ality of $R$ is 0.
\item The non-nearest-neighbor interaction term $\mathrm{Re\,}\chi_R(U_xU^\dagger_{x+2\hat\mu})$ 
can be added to the action if the $N$-ality of $R$ is 0 and the coefficient in front of it is positive. 
(If the distance is larger than two lattice spacings, or if the coefficient is negative, then 
the site-reflection positivity becomes hard to prove.)
\item The non-nearest-neighbor interaction term $\mathrm{Re\,}\chi_R(U_xU^\dagger_{x+\hat\mu+\hat\nu})$, $\mu\neq\nu$, 
can be added to the action if the $N$-ality of $R$ is 0.
\end{enumerate}
\hfil\hspace{80pt}*\hspace{80pt}*\hspace{80pt}*\hspace{80pt}*\hspace{80pt}

%
%
%
Note that our proved inequality may fail to be useful in phases other than the disordered phase, 
even though it is correct in any phases. 
Consider a spin system with a global symmetry group $G$ 
in a $d$-dimensional box of size $L_1\times \dots\times L_d$ 
whose boundary condition is twisted by $g\in G$ in the $x^1$-direction and otherwise periodic. 
Let $Z^g$ denote the twisted partition function and set 
$\displaystyle L_\perp\equiv \prod_{k=2}^{d}L_k$. 
Based on our experience in gauge theories, we generally expect following behaviors of $Z^g/Z$ 
in various phases:
\begin{e}
\begin{array}{lr}
Z^{g}/Z\approx \exp(-xL_\perp\exp(-y L_1))&\textrm{(Disordered phase)}
\\
Z^{g}/Z\approx \exp(-zL_\perp)&\textrm{(Ordered phase)}
\\
Z^{g}/Z\approx \exp(-wL_\perp/L_1)&\textrm{(Massless phase)}
\end{array}
\label{anticipation}
\end{e}
\label{anticipation_}
where $x,y,z$ and $w$ are functions of the coupling constants and the choice of $g$. 
Letting $\alpha$ denote the value of $(1-Z^{g}/Z)^{1/L_1}$ in the thermodynamic limit, 
we have $\alpha=\ee^{-y}$ in the disordered phase, while $\alpha=1$ in the other cases 
so that the r.h.s. of (\ref{problem}) tends to a constant independent of $n$. 
So the bottom line is that algebraic decay of correlation function cannot be inferred 
from the behavior of the r.h.s. in general.

Finally we comment on the formal 
difference between the inequality derived by ${\rm{Kov\acute{a}cs}}$ and Tomboulis 
in refs.\cite{spin-KT,spin-K} and ours in the two-dimensional $SU(2)\times SU(2)$ PCM. Their result is
\begin{e}
\langle\Gamma_\mu(x;n)\rangle\Big|_{L_\mu=\infty}
\leq\jfrac{Z_n(+,+)-Z_n(-,-)}{Z_n(+,+)+Z_n(+,-)+Z_n(-,+)+Z_n(-,-)},
\end{e}
where $Z_n(\tau_1,\tau_2)\ (\tau_{1,2}=\pm 1)$ is the partition function on the lattice of size $n\times n$ with 
a twist $\tau_\mu$ operated in the $x^\mu$-direction.

On the other hand, our result ((\ref{problem}) with $G=Z_2$) gives
\begin{align}
\langle\Gamma_\mu(x;n)\rangle&\leq 2\left\{1-\langle\mathcal F_0[\Nu]\rangle\right\}^{n/L_\mu}
\\
&= 2\left\{\jfrac{1}{2}\Bigg[\jfrac{Z_{L_\mu}(+,+)-Z_{L_\mu}(+,-)}{Z_{L_\mu}(+,+)}\Bigg]\right\}^{n/L_\mu}.
\end{align}
Note that in the latter, both r.h.s. and l.h.s. are estimated on the lattice of size $L_\mu$. 
In both formulas the correlation function is assumed to be in the fundamental representation. Although 
they look different, both relates the exponential suppression of 
the wall free energy (in the thermodynamic limit) to the mass gap, thus their physical contents are 
totally consistent.

\subsection{Demonstration in 1D PCM and 2D square Ising model}
Let us verify the proved inequality (\ref{problem}) 
explicitly in the $G\times G$ PCM in one-dimension, with $G$ an arbitrary compact group. 
The partition function of the model on a periodic chain of length $L$ is given by
\begin{align}
Z_\Lambda&\equiv\int_{}^{}\prod_{k=1}^{L}dU_k\ \exp\Big(\beta\jsum{i=1}{L}{\mathrm{Re\ Tr\ }}(U_iU_{i+1}^{-1})\Big),
\hspace{40pt}\beta>0,\ U_i\in G,
\\
&=\int_{}^{}\prod_{k=1}^{L}dU_k\ \prod_{i=1}^{L}\Big[\jsum{r}{}d_rF_r\chi_r(U_iU_{i+1}^{-1})\Big],
\end{align}
where $\jsum{r}{}$ runs over all irreducible unitary representations of $G$. 
$F_r=F_{\overline{r}}>0$ follows from 
the reflection positivity and reality of the action. Let us define $c_r\equiv F_r/F_0$ for later convenience 
($0$ denotes the trivial representation). Straightforward calculation yields
\begin{e}
Z_\Lambda=\jsum{r}{}d_r^2(F_r)^L.
\end{e}
The twisted partition function $Z_\Lambda^g$ is similarly given by
\begin{align}
Z_\Lambda^g&\equiv\int_{}^{}\prod_{k=1}^{L}dU_k\ \Big[\jsum{r'}{}d_{r'}F_{r'}\chi_{r'}(gU_1U_{2}^{-1})\Big]
\prod_{i=2}^{L}\Big[\jsum{r}{}d_rF_r\chi_r(U_iU_{i+1}^{-1})\Big],   \hspace{40pt}   g\in G',
\\
&=\jsum{r}{}d_r\chi_r(g)(F_r)^L,
\end{align}
where $G'\subset G$ is an arbitrary subgroup of $G$. 
($Z_\Lambda^g$ reduces to $Z_\Lambda$ for $g=\mathbf{1}$, as it should be.) Hence we get
\begin{align}
\jlim{L\to\infty}\Big\{1-\int_{G'}^{}dg\langle \mathcal O(g)\rangle\Big\}^{n/L}
&=\jlim{L\to\infty}\left\{\jfrac{{\jsum{r}{}}'d_r^2(c_r)^L}{\jsum{r}{}d_r^2(c_r)^L}\right\}^{n/L}
\\
&=(c_{r'})^n.   \label{TA}
\end{align}
Here ${\jsum{r}{}}'$ is defined as a sum over all representations of $G$ \textit{which are nontrivial w.r.t.} $G'$, 
and $c_{r'}$ is defined as the largest one among $\{c_r\ |\ r\ {\textrm{is\ nontrivial\ w.r.t.\ }G'}\}$.

Next we define the correlation function as $\Gamma(n)=\chi_R(U_0U_n^{-1})$. Then (\ref{technical}) requires $R$ to be nontrivial 
w.r.t. $G$. After straightforward calculation we get
\begin{e}
\jfrac{\langle\Gamma(n)\rangle}{\langle\Gamma(0)\rangle}=\jfrac{1}{d_R}\langle\chi_R(U_0U_n^{-1})\rangle=(c_R)^n   \label{TB}
\end{e}
in the thermodynamic limit ($L\to\infty$). Since $c_R\leq c_{r'}$ is obvious from their definitions, we conclude from 
(\ref{TA}) and (\ref{TB}) that 
the inequality (\ref{problem}) certainly holds at least in the limit $L\to\infty$.
\footnote{(\ref{problem}) should hold for finite $L$ too, but 
expressions of both sides of (\ref{problem}) become highly complicated for finite $L$ and 
verification seems to be hard.}
\vspace{5pt}
\\
\hfil\hspace{80pt}*\hspace{80pt}*\hspace{80pt}*\hspace{80pt}*\hspace{80pt}

As a next example let us take the two-dimensional Ising model on a square lattice. 
The partition function of the model is given by
\begin{e}
Z_\Lambda=\int_{}^{}\prod_{k\in\Lambda} d\sigma_k\ 
\exp\left(\jsum{\mu=1}{L_1}\jsum{\nu=1}{L_2}(a\sigma_{\mu\nu}\sigma_{\mu+1,\nu}+b\sigma_{\mu\nu}\sigma_{\mu,\nu+1})\right),
\end{e}
where $\sigma_{\mu\nu}$ is the Ising spin located at the site $(\mu,\nu)$ and 
$\jint{}{}d\sigma\equiv\jfrac{1}{2}\jsum{\sigma=\pm 1}{}$. 
Periodic boundary conditions are imposed so that 
$\sigma_{1,\nu}=\sigma_{L_1+1,\nu},\,\sigma_{\mu,1}=\sigma_{\mu,L_2+1}$. 
We assume $a>0,\,b>0$. 
Let us focus on the high temperature (disorder) phase of the model.

The exact asymptotic form of the two-point correlation function is known \cite{textbooks} 
and the mass gap (or inverse correlation length) 
$M\equiv 2(\oline{a}-b)$, where $\oline{a}$ is the \textit{dual temperature} defined by
\begin{e}
\sinh 2a\ \sinh 2\oline{a}=1.
\end{e}
$\oline{b}$ is defined in the same way.

To estimate the free energy of walls we need explicit formulae for twisted and untwisted 
partition functions. 
Here we use the expressions due to Kastening \cite{Kastening}, which in out notation read
\begin{align}
&Z_\Lambda=\jfrac{1}{2}[2\sinh(2a)]^{L_1L_2/2}\times     \notag
\\
&
\left\{
\prod_{k=1}^{L_2}\Big[2\cosh\Big(\jfrac{L_1}{2}\gamma_{2k-1}\Big)\Big]
+
\prod_{k=1}^{L_2}\Big[2\sinh\Big(\jfrac{L_1}{2}\gamma_{2k-1}\Big)\Big]
+
\prod_{k=1}^{L_2}\Big[2\cosh\Big(\jfrac{L_1}{2}\gamma_{2k-2}\Big)\Big]
-
\prod_{k=1}^{L_2}\Big[2\sinh\Big(\jfrac{L_1}{2}\gamma_{2k-2}\Big)\Big]
\right\},
\\
&Z_\Lambda^{(-)}=\jfrac{1}{2}[2\sinh(2a)]^{L_1L_2/2}\times    \notag
\\
&
\left\{
\prod_{k=1}^{L_2}\Big[2\cosh\Big(\jfrac{L_1}{2}\gamma_{2k-1}\Big)\Big]
+
\prod_{k=1}^{L_2}\Big[2\sinh\Big(\jfrac{L_1}{2}\gamma_{2k-1}\Big)\Big]
-
\prod_{k=1}^{L_2}\Big[2\cosh\Big(\jfrac{L_1}{2}\gamma_{2k-2}\Big)\Big]
+
\prod_{k=1}^{L_2}\Big[2\sinh\Big(\jfrac{L_1}{2}\gamma_{2k-2}\Big)\Big]
\right\}.
\end{align}
$Z^{(-)}_\Lambda$ is the twisted partition function; more precisely, it is a partition function 
on a lattice which is antiperiodic in $x^1$-direction and periodic in $x^2$-direction. (Note that 
this boundary condition is equivalent to the existence of a wall wrapping around a periodic lattice 
in $x^2$-direction.) $\gamma_k>0$ is defined by
\begin{e}
\cosh \gamma_k=\cosh 2\oline{a}\ \cosh 2b-\cos\jfrac{\pi k}{L_2}
\sinh 2\oline{a}\ \sinh 2b.
\end{e}
The inequality to be checked, namely (\ref{problem}) for the square Ising model, is given by
\begin{e}
\langle\sigma_0\sigma_n\rangle_\Lambda
\leq 2\left\{\jfrac{1}{2}\left(1-\jfrac{Z_\Lambda^{(-)}}{Z_\Lambda}\right)\right\}^{n/L_1}.
\label{-153}
\end{e}\label{U}
$\langle\dots\rangle_\Lambda$ denotes the expectation value measured on a finite lattice ($=\Lambda$). 
Similarly $\langle\dots\rangle_\infty$ denotes an expectation value in the thermodynamic limit. 
For simplicity we calculate not $\jfrac{Z^{(-)}_\Lambda}{Z_\Lambda}$ but
\begin{align}
\jfrac{Z_\Lambda-Z^{(-)}_\Lambda}{Z_\Lambda+Z^{(-)}_\Lambda}
&=\jfrac{1-\displaystyle\prod_{k=1}^{L_2}\tanh\Big(\jfrac{L_1}{2}\gamma_{2k-2}\Big)}
{\displaystyle
\prod_{k=1}^{L_2}\jfrac{\cosh\Big(\jfrac{L_1}{2}\gamma_{2k-1}\Big)}{\cosh\Big(\jfrac{L_1}{2}\gamma_{2k-2}\Big)}+
\displaystyle
\prod_{k=1}^{L_2}\jfrac{\sinh\Big(\jfrac{L_1}{2}\gamma_{2k-1}\Big)}{\cosh\Big(\jfrac{L_1}{2}\gamma_{2k-2}\Big)}
}.  \label{Z@}
\end{align}
Considering that\footnote{$0<\gamma_0$ stems from the fact that now the system is in the 
high temperature (disorder) phase.}
 $0<\gamma_0=2(\oline{a}-b)$ is the smallest among $\{\gamma_k\}$, we obtain, after some algebra,
\begin{e}
\lim_{L_1\to\infty}\left(1-\frac{Z_\Lambda^{(-)}}{Z_\Lambda}\right)^{1/L_1}
=\exp\Big\{-\Big(\gamma_0+\jfrac{1}{2}\jsum{k=0}{2L_2-1}(-1)^{k+1}\gamma_k\Big)\Big\}.
\label{-159}
\end{e}
Since $\jsum{k=0}{2L_2-1}(-1)^{k+1}\gamma_k=O\Big(\jfrac{1}{L_2}\Big)$ for $L_2\gg 1$, we get
\begin{e}
\lim_{L_2\to\infty}\lim_{L_1\to\infty}
\Big[{\rm{r.h.s.\ of\ }}(\ref{-153})\Big]=2\ee^{-\gamma_0n}.\label{777771}
\end{e}
We compare this result with the asymptotic form of the exact two point function in the high temperature 
phase \cite{textbooks}:
\begin{e}
\hspace{50pt}\langle\sigma_0\sigma_n\rangle_\infty=f(a,b)\jfrac{1}{\sqrt{\mathstrut n}}\ee^{-\gamma_0n}
\times\Big[1+O\big(\jfrac{1}{n}\big)\Big]
\hspace{50pt}{\rm{for}}\ \ n\gg 1,      \label{exact_two_point}
\end{e}
where the factor $f(a,b)$ is independent of $n$.

(\ref{777771}) and (\ref{exact_two_point}) tells that 
the exponential decay rates of both sides \textit{coincide exactly} for all values of 
$(a,\,b)$ when the system is in the disorder phase. This result suggests that 
our inequality might be a rather accurate one in general.

\subsection{Demonstration in 2D triangular Ising model}\label{3-5}
Our next example is the two-dimensional Ising model on a \textit{triangular} lattice. 
In the square Ising model, exact expression for mass gap was already known, thus the value of 
our theorem is obscured. However, as for the triangular Ising model, 
an exact expression for mass gap is not known except for special cases, so (unlike in the previous section) 
\textit{the results we give in this section are essentially new.}

We consider a lattice of size $L_1\times L_2$ with periodic boundary conditions.\footnote{$L_1$ denotes 
the number of triangles. It is not the \textit{actual length} of the lattice.} 
See fig.\ref{tri-1} for 
an example with $L_1=6$ and $L_2=3$; upper and lower edges painted in blue should be identified, and also 
right and left edges painted in red should be identified. (This lattice is the one that appeared in 
the seminal work of Houtappel \cite{Houtappel} in which an analytic formula for the triangular Ising model 
was obtained for the first time.) 
Another triangular lattice commonly used in the 
literature is depicted in fig.\ref{tri-2}, where edges are again colored for the purpose of indicating 
the periodic structure of the lattice. It is easy to prove that these lattices are equivalent 
\textit{if and only if} $L_1$ is a multiple of $2L_2$. See fig.\ref{tri-3} for illustration of this fact. 
Numbers are written to guide the eye; edges assigned with the same number should be identified. 
Hereafter we will assume this condition, but this is only a technical assumption and not essential for we will be 
interested in the limit $L_1\to\infty$. 
The reason we did not start with the lattice in fig.\ref{tri-2} is because it does not allow for simple use of 
reflection positivity.
\begin{figure}[phbt]
\begin{minipage}{.46\textwidth}
\begin{center}
\includegraphics*[width=6.5cm]{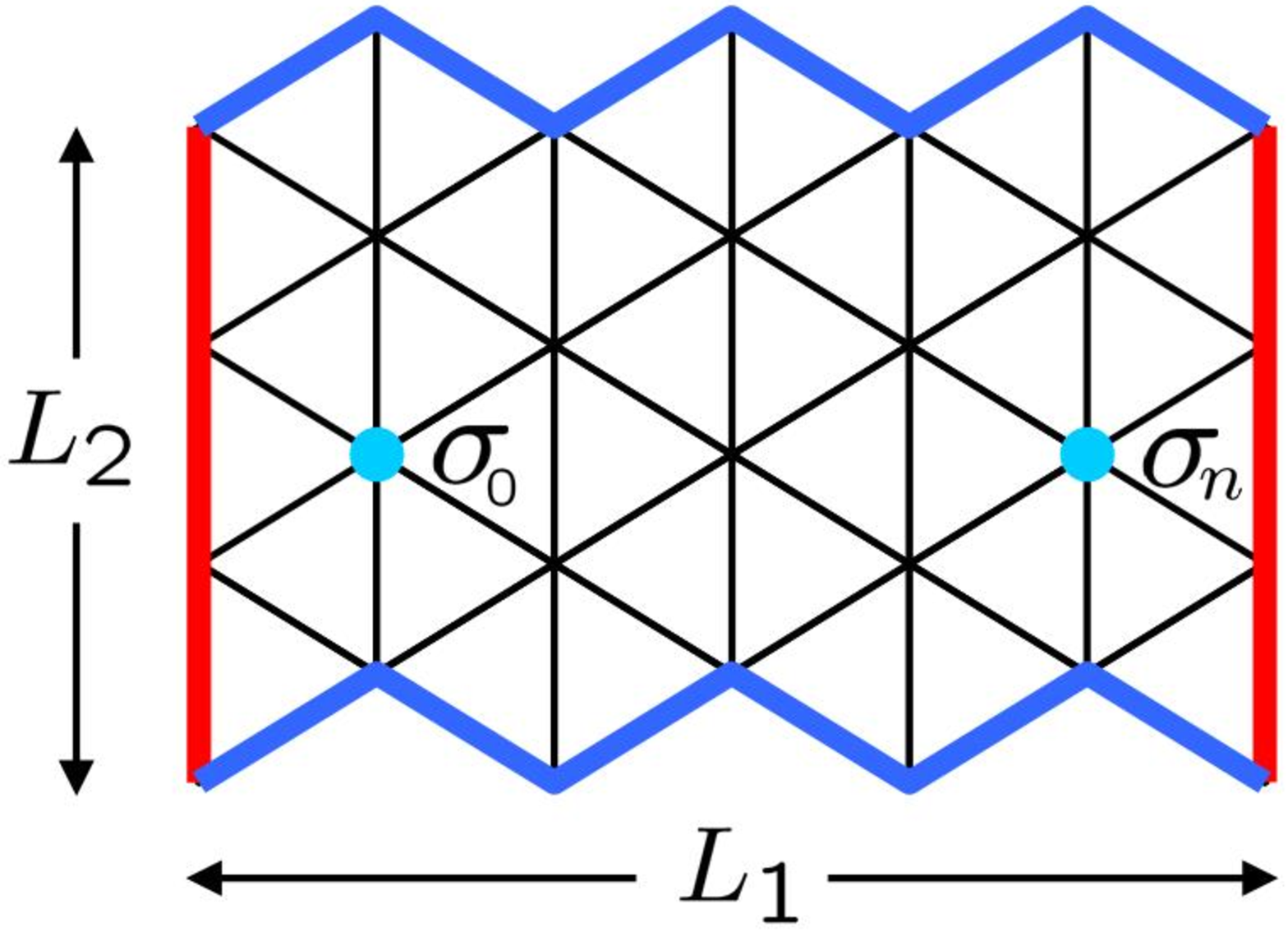}
\end{center}\vspace{-16pt}
\caption{A triangular lattice of size $L_1\times L_2$ with periodic boundary conditions; 
the red and the blue ends are identified, respectively. This lattice is symmetric w.r.t. each of vertical axes.}
\label{tri-1}
\end{minipage}\hfill
\begin{minipage}{.46\textwidth}
\begin{center}
\includegraphics*[width=6.5cm]{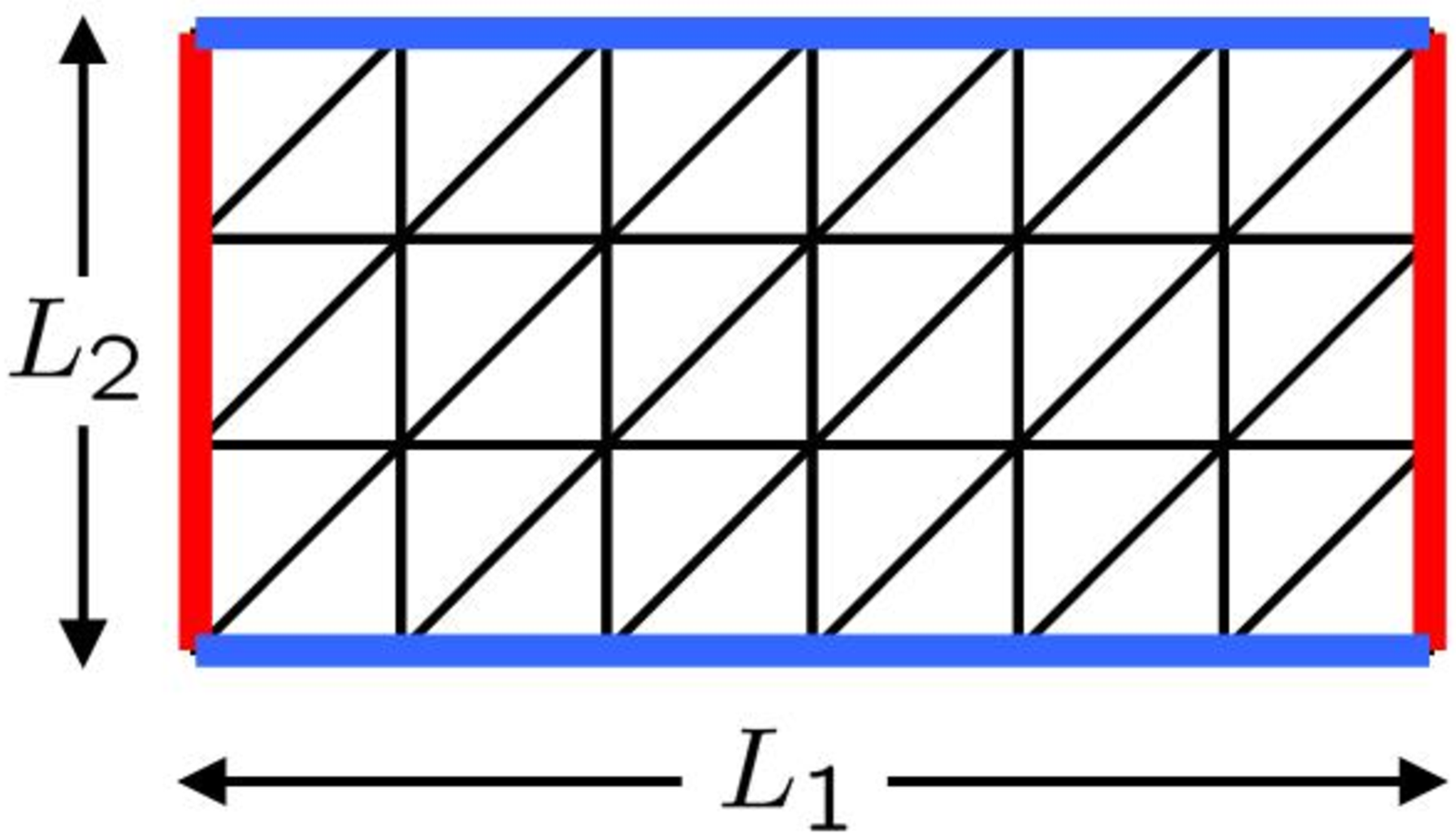}
\end{center}\vspace{-16pt}
\caption{A triangular lattice constructed from a square lattice by addition of diagonal edges. Its 
periodic structure is indicated by coloring as in fig.\ref{tri-1}.}
\label{tri-2}
\end{minipage}
\end{figure}
\begin{figure}[hbtp]
\begin{minipage}{.46\textwidth}
\begin{center}
\includegraphics*[width=5.5cm]{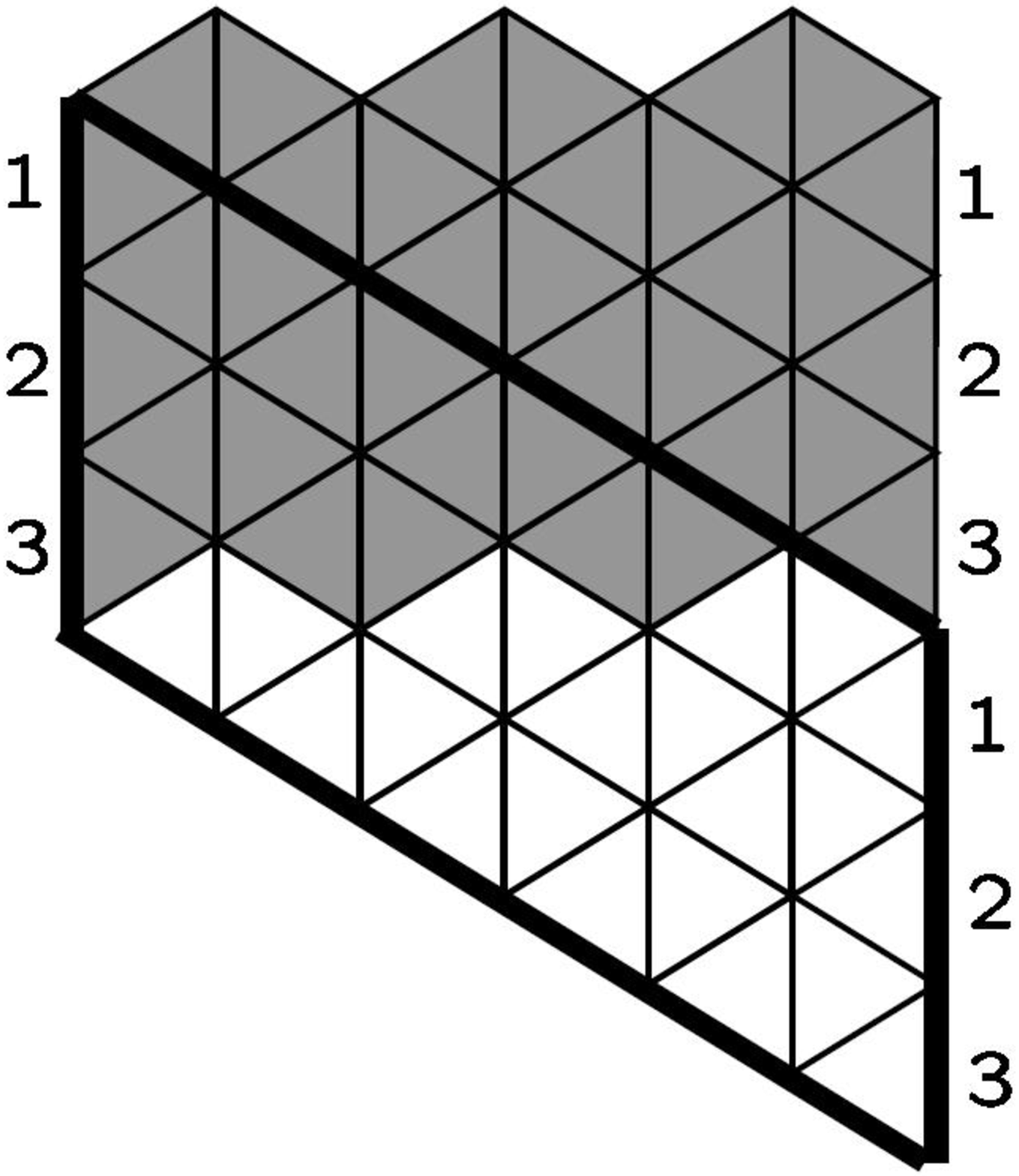}
\end{center}\vspace{-20pt}
\caption{An illustration of the fact that those lattices given in figs.\ref{tri-1} and \ref{tri-2} are 
equivalent iff $L_1$ is a multiple of $2L_2$. Here $L_1=6$ and $L_2=3$.}
\label{tri-3}
\end{minipage}
\hfill 
\begin{minipage}{.46\textwidth}
\begin{center}
\includegraphics*[width=6.5cm]{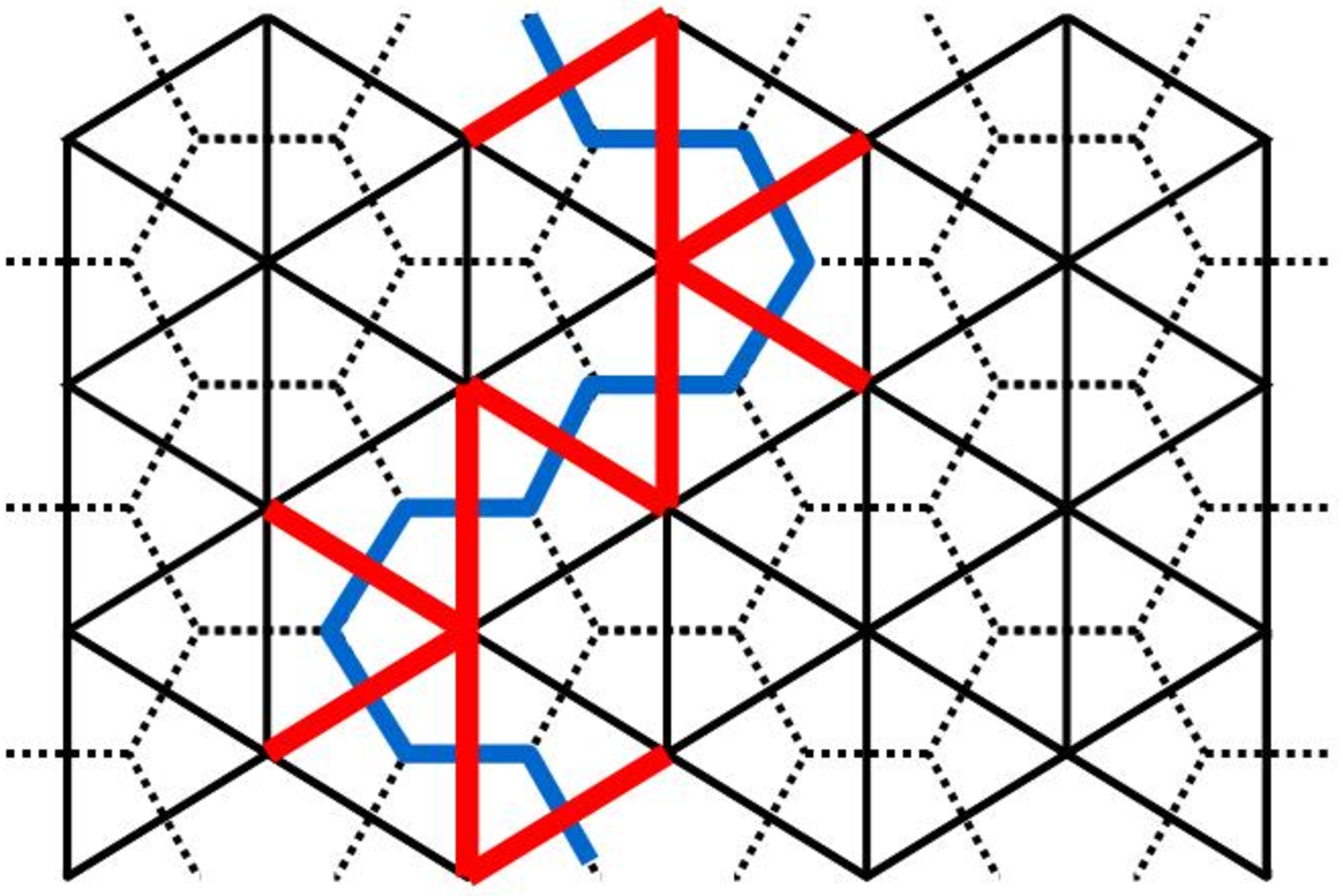}
\end{center}\vspace{-16pt}
\caption{A $Z_2$-twist represented by a blue loop on the 
dual (hexagonal) lattice. On the original (triangular) lattice, it is represented by a 
stacked set of links (colored in red) with couplings of opposite sign. }
\label{tri-4}
\end{minipage}
\end{figure}

Let us define a twist on a planar triangular lattice. A twist is a closed loop on the dual lattice, and 
the dual of a triangular lattice is a honeycomb (or hexagonal) lattice as shown in fig.\ref{tri-4}. 
Note that the blue line in fig.\ref{tri-4} is a closed loop owing to the periodic structure of the lattice. 
It becomes, on the original lattice, a stacked set of links with a coupling constant of opposite sign, 
which is depicted as a set of red links in fig.\ref{tri-4}. Note that introducing a twist to a periodic lattice 
as in fig.\ref{tri-4} is equivalent to imposing an anti-periodic boundary condition in the horizontal direction.

Let us remember that it is not the number of the walls 
but rather the number \textit{mod 2} of them that is physically relevant. This `$Z_2$ conservation' of walls is 
a direct consequence of $\sigma^2=1$ in the present model, and we can show it explicitly 
by a sequence of changes of variables $\sigma\to -\sigma$. 
For instance, the partition function containing one wall and that containing three walls agree 
completely as illustrated in fig.\ref{Z=Z} in which red segments represent twisted links.

\begin{figure}[hbtp]
\begin{center}
\includegraphics*[width=11.0cm]{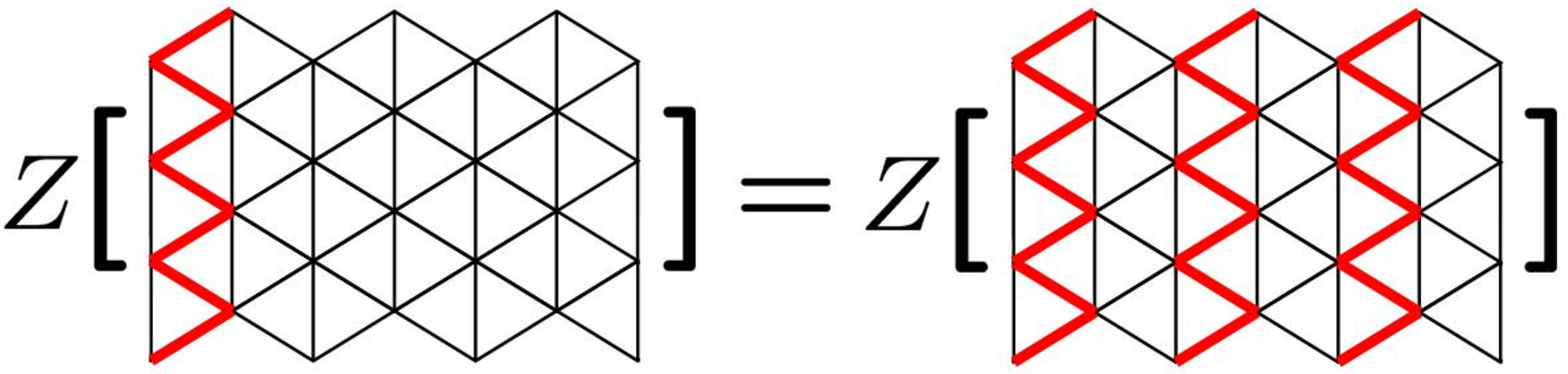}
\end{center}\vspace{-10pt}
\caption{The partition function does not differ for any odd number of twists, owing to the 
$Z_2$ conservation of the twist.}
\label{Z=Z}
\end{figure}
Exact expressions for partition functions of a planar triangular Ising model with various boundary conditions 
were derived by Wu and Hu via `Grassmann path integral method' \cite{Wu-Hu}. 
Let $Z_\Lambda$ ($Z^{(-)}_\Lambda$) denote the partition function with periodic boundary condition 
in both directions (with periodic in vertical and anti-periodic in horizontal direction), respectively. 
On a triangular lattice, three different couplings can be defined in each directions, so let us introduce 
$J_1$ as the coupling constant on vertical bonds in figs.\ref{tri-1},\ref{tri-2} and $J_2, J_3$ the other two. 
Reflection positivity however requires $J_2=J_3\,(\equiv J)$. Introduce
\begin{e}
t_1\equiv \tanh(J_1/{\rm{k_B}}T),\ \ t\equiv \tanh(J/{\rm{k_B}}T).
\end{e}
Our convention is such that $t_{(1)}>0$ corresponds to ferromagnetic coupling. On the $(t,t_1)$-plane, 
there is a line which corresponds to $T=T_c$ and we will call it the ``critical line'' in the following.
Under the change of notation $L_1\to N$ and $L_2\to M$, the result due to Hu and Wu 
for this case reads
\begin{align}
Z_\Lambda&=\jfrac{1}{2}\left[2\cosh^3(\beta J)\right]^{MN}\Big[\Omega_{\frac{1}{2},\frac{1}{2}}+\Omega_{\frac{1}{2},0}
+\Omega_{0,\frac{1}{2}}-\textrm{sgn}\,(T-T_c)\Omega_{0,0}\Big],\label{/163}
\\
Z_\Lambda^{(-)}&=\jfrac{1}{2}\left[2\cosh^3(\beta J)\right]^{MN}\Big[\Omega_{\frac{1}{2},\frac{1}{2}}+\Omega_{\frac{1}{2},0}
-\Omega_{0,\frac{1}{2}}+\textrm{sgn}\,(T-T_c)\Omega_{0,0}\Big],\label{/164}
\end{align}
where
\begin{e}
\Omega_{\mu\nu}=(A_0)^{MN/2}\prod_{p=0}^{M-1}\prod_{q=0}^{N-1}
\Big[1-B\cos\jfrac{2\pi(p+\mu)}{M}-A\cos\jfrac{2\pi(q+\nu)}{N}
-A\cos\Big(\jfrac{2\pi(p+\mu)}{M}-\jfrac{2\pi(q+\nu)}{N}\Big)\Big]^{1/2},
\end{e}
\begin{e}
A_0=(1+t^2t_1)^2+(t_1+t^2)^2+2t^2(1+t_1)^2,
\ \ \ 
A=\frac{2(1-t_1^2)(1-t^2)t}{A_0},
\ \ \ 
B=\frac{2t_1(1-t^2)^2}{A_0}    \label{A_def}
\end{e}
with $T_c$ the phase-transition temperature. 
Using the formulae given above, we can show
\begin{lem}\label{MNlimit}
In the disordered phase $(\ \Leftrightarrow T>T_c\ )$ we have
\begin{e}
\lim_{M\to\infty}\lim_{N\to\infty}\left(1-\frac{Z^{(-)}_\Lambda}{Z_\Lambda}\right)^{1/N}=\ee^{-\rho},
\label{kakkiteki}
\end{e}
where
\begin{gather}
\rho\equiv \cosh^{-1}\left(\frac{g(B)}{|A|}\right)>0,\label{definition_of_rho}
\\
g(x)\equiv\left\{\begin{matrix}
\sqrt{-2x(1+x)}&\displaystyle\left(-1<x<-\frac{1}{3}\right)
\\
\displaystyle\frac{1-x}{2}&\displaystyle\left(-\frac{1}{3}\leq x<1\right)
\end{matrix}
\right..\label{definition_of_g}
\end{gather}
The order of two limits in (\ref{kakkiteki}) must not be changed.
\end{lem}
In the above, $\displaystyle\frac{g(B)}{|A|}\geq 1$ and $|B|<1$ are implicitly assumed; these can be shown 
for every $(t,t_1)\in(-1,1)^2$ by elementary methods. (Note that $\displaystyle\frac{g(B)}{|A|}=1$ defines 
the critical line.) 
The proof of theorem \ref{MNlimit} is elementary but technically cumbersome, which we relegate to the appendix. 

To gain an intuitive understanding of the above result, let us 
see fig.\ref{pm3d_mass_gap}, in which the projection of $1-\ee^{-\rho}$ onto the $(t_1,t)$-plane is drawn. 
The black region corresponds 
to the (anti-)ferromagnetically ordered phase, while the colored region to the disordered phase. 
Brighter color represents larger $\rho$, hence larger mass gap.
\begin{figure}[btph]
\begin{center}
\includegraphics*[width=15.0cm]{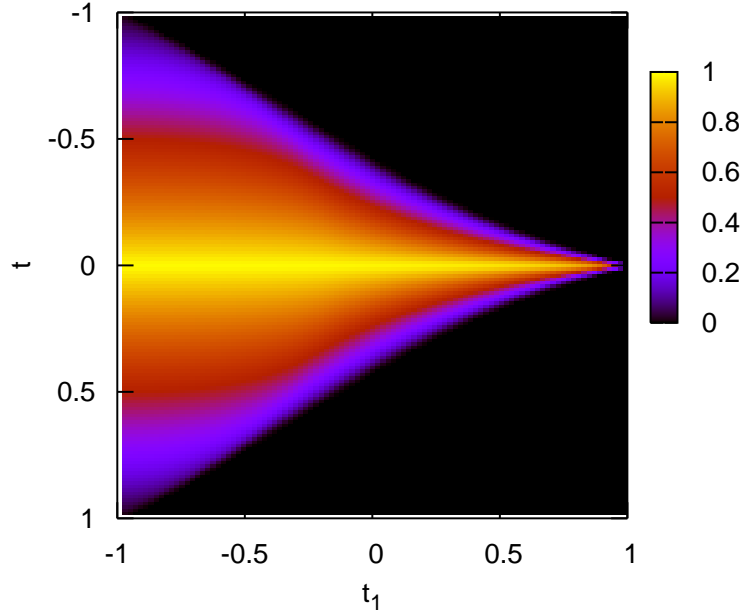}
\end{center}\vspace{-40pt}
\caption{The projection of $1-\ee^{-\rho}$ onto the $(t_1,t)$-plane. 
The black region corresponds to the (anti-)ferromagnetically ordered phase 
and the colored region to the disordered phase. 
Brighter (darker) color represents larger (smaller) $\rho$ and especially $\rho$ diverges on the $t=0$ line. 
The boundary of the colored region signifies the critical line.}
\label{pm3d_mass_gap}
\end{figure}
Fig.\ref{pm3d_mass_gap} clearly shows a symmetry under $t\leftrightarrow -t$; this is a manifestation of the well-known 
fact that the triangular Ising model is invariant under simultaneous sign reversal of any two of $J_1,J_2,J_3$.

The inequality of our primary interest, namely (\ref{problem}) for the triangular Ising model, reads\footnote{
Actually, we originally proved (\ref{problem}) on a \textit{square} lattice, but 
the whole procedure of the proof goes over to 
the case of a \textit{triangular} lattice almost unchanged.}
\begin{e}
\langle\sigma_0\sigma_n\rangle_\Lambda
\leq 2\left\{\jfrac{1}{2}\left(1-\jfrac{Z_\Lambda^{(-)}}{Z_\Lambda}\right)\right\}^{n/N}.
\label{nobita}
\end{e}
Letting $M\to\infty$ after $N\to\infty$, we obtain
\begin{thm}
\begin{e}
\langle\sigma_0\sigma_n\rangle_\infty \leq 2\ee^{-\rho n}.
\end{e}
\vspace{-25pt}
\end{thm}
The above is the main result in this subsection; $\rho$ is a rigorous lower bound of the true mass gap. 
We should keep in mind that the l.h.s. of (\ref{nobita}) is a pair correlation between two spins 
on the same horizontal level as depicted in fig.\ref{tri-1}; the two spins are \textit{not} on the same 
lattice axis.

Since the exponential decay rates are equal for both sides of the inequality in the the square Ising model, 
it is natural to expect so in the triangular Ising model too. 
The asymptotic correlation between two spins \textit{on the same lattice axis} in 
the triangular Ising model was derived by Stephenson for both ferromagnetic and antiferromagnetic 
couplings \cite{Stephenson}. 
However, the asymptotic correlation between two spins \textit{off the axis} is not found in the literature. 
We conjecture as follows:
\vspace{10pt}
\hrule
\vspace{10pt}
\noindent\textbf{Conjecture.}
\\\hfil 
$\rho$ is equal to the true off-axis mass gap for every $(t,t_1)$ in the disordered phase.
\vspace{10pt}
\hrule
\vspace{10pt}
\noindent 
The most straightforward way to test the conjecture would be to measure the mass gap directly 
via Monte Carlo simulation. However, as already seen from (\ref{exact_two_point}) the exponential 
falloff generically receives power law corrections (the so-called `Ornstein-Zernike' decay \cite{OZ-decay}) 
which makes a reliable fitting difficult. 
To evade this hamper would call for sophisticated methods such as the Monte Carlo Transfer Matrix 
Method \cite{MC-transfer-matrix-method}. A numerical check of the conjecture therefore seems to be a highly 
nontrivial task, and we defer it to future work.

In a special case, analytical test is possible: 
when $t_1=0$ the model reduces to the isotropic square Ising model and the off-axis 
correlation function reduces to 
the \textit{diagonal} correlation function. From (\ref{definition_of_rho}) it follows that
\begin{e}
\rho\big|_{t_1=0}=\cosh^{-1}\frac{(1+t^2)^2}{4|t|(1-t^2)}\ \ \ \ \ 
\textrm{for}\ \ |t|<t_C=\sqrt{2\mathstrut}-1,		\label{diagonal_CW}
\end{e}
which completely agrees with the exact \textit{diagonal} mass gap 
obtained by Cheng and Wu in 1967 \cite{textbooks}.

Further insight is gained by considering the isotropic case $t=t_1$. 
Since the two-point correlation function in this case is expected to be approximately isotropic 
(except for sign in antiferromagnetic case), 
it seems reasonable to compare $\ee^{-\rho}$ with $\ee^{-(\sqrt{3}/2)m}$ where $m$ is the exact 
\textit{on-axis} mass gap \cite{Stephenson}. ($\sqrt{3\mathstrut}/2$ is a geometric correction factor.)

\begin{figure}[bthp]
\begin{center}
\includegraphics*[width=11.0cm]{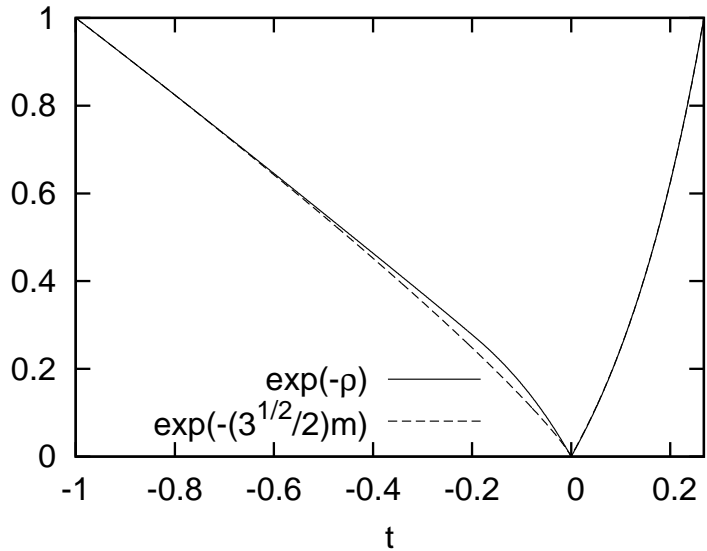}
\end{center}\vspace{-20pt}
\caption{$\ee^{-\rho}$ and $\ee^{-(\sqrt{3}/2)m}$ are plotted against $t$ in the case of isotropic case ($t=t_1$). 
Their agreement is remarkable especially at $t$ positive.}
\label{Stephenson}
\end{figure}

Fig.\ref{Stephenson} depicts the graphs of $\ee^{-\rho}$ and $\ee^{-(\sqrt{3}/2)m}$ 
against $t\in[-1,\,2-\sqrt{3\mathstrut}\,]$. For $t>0$ they agree quite well; 
their nonzero difference is hardly discernible to the eye. For $t<0$ agreement is still not bad. 
To say the least, the comparison suggests that $\rho$ be fairly close to the true off-axis mass gap and 
supports, rather than defies, the conjecture. 

It is readily seen from (\ref{definition_of_g}),(\ref{definition_of_rho}) that $\ee^{-\rho}$ is 
a nonanalytic function in the region of negative $t_1$, 
which, assuming the validity of our conjecture, 
implies non-analyticity of the mass gap. Such an exotic possibility definitely deserves further study. 
Strictly speaking, however, there are different possibilities that cannot be denied here: 
for example it could be the case that 
$\rho$ equals the true mass gap \textit{only when} $t_1\geq 0$. In the latter case, non-analyticity of 
$\rho$ does not signify that of the true mass gap.

Let us end this subsection by invoking the effectiveness of our approach. 
Although the circumstance concerning our conjecture is rather moot, 
it can be safely said that our result in this subsection is essentially new to the extent that it 
rigorously gives a lower bound for the still-unknown off-axis mass gap of the triangular (both 
ferromagnetic and antiferromagnetic, both isotropic and anisotropic) Ising model in the disordered phase.
\subsection{Strong coupling analysis}\label{SC_expansion}
In this section we show, using the convergent strong-coupling (taken as synonymous with high temperature) expansion, that 
both sides of the proved inequality (\ref{problem}) 
have an identical exponential decay rate at long distance as long as the on-axis correlation function is considered.\footnote{
High-temperature behavior of correlation functions 
in Ising-like models have been studied by many authors in a variety of methods; see ref.\cite{mass_gap}, for example. 
It is worthwhile to note that a majority of existing studies deal with 
neither the off-axis correlation function nor the case of an antiferromagnetic coupling. Hopefully a partial 
understanding of this fact will be gained through the discusions in this subsection.}
The proof is valid in any dimension and makes no use of 
reflection positivity. Since the corresponding result in LGT has already been derived by $\rm{M\ddot unster}$ 
\cite{Munster} (as mentioned in section \ref{Munster_}) and since no essential difficulty arises in extending his proof 
to the case of spin models, we shall be brief here and only try to sketch the main idea behind the approach. Implications of 
this result to our conjecture will be discussed later.

For simplicity of exposition let us consider the isotropic 
\underline{square} Ising model and its on-axis correlation function (though 
our argument is readily extendable to more general models such as PCM). A precise statement of the claim goes 
follows: as long as the size of the lattice is larger enough than $n$, the strong coupling expansion 
(SCE) of $\rho$ and $m$ are identical at least up to order $n$. (Our notation is such that the definition of 
$\rho$ is in (\ref{kakkiteki}), $m$ is the mass gap, $Z,Z^{(-)}, L_1,L_2$ are the same as in section \ref{U} and 
$t\equiv\tanh(J/{\rm{k_B}}T)$.) 

Let us begin with the expression $\displaystyle Z=\sum_{\{\sigma\}}^{}\prod_{i\not=j}^{}(1+t\sigma_i\sigma_j)$. 
Expanding $Z$ into sums of disconnected loops and then taking the logarithm, we have $\displaystyle
\log Z=\sum_{\gamma}^{}t^{|\gamma|}$ with $\gamma$ any connected loop and $|\gamma|$ the perimeter of $\gamma$. 
Using similar expression for $Z^{(-)}$ we obtain 
$\displaystyle\log\frac{Z^{(-)}}{Z}=-2\sum_{\gamma\in S}^{}t^{|\gamma|}$ where $S$ is the set of loops wrapping 
around the lattice in $x^1$-direction for odd number of times. Since we are interested in the limit $|t|\ll 1$, 
it is sufficient to consider only such loops that wind around the lattice in $x^1$-direction only once. Factorizing 
the degeneracy factor due to translational symmetry in $x^2$-direction, we have 
$\displaystyle\log\frac{Z^{(-)}}{Z}=-2L_2\sum_{\gamma\in S'}^{}t^{|\gamma|}$; the definition of $S'$ should be obvious.

It is clear that the leading contribution, of order $O(t^{L_1})$, comes from a straight line extending in $x^1$-direction while 
the subleading contributions come from loops which are formed via addition of some `decorations' to the leading line. 
Dividing by the leading contribution and 
taking the logarithm will single out contributions of connected decorations, 
which is proportional to $L_1$ owing to the translational invariance of the straight line. Thus we find exactly the behavior 
(\ref{anticipation}) in section \ref{anticipation_}:
\begin{e}
\log\Big[\left(\frac{1}{L_2}\log\frac{Z^{(-)}}{Z}\right)\Big/t^{L_1}\Big]\propto L_1.\label{prop}
\end{e}
On the other hand, the on-axis two-point correlation function $\langle \sigma_{x}\sigma_{x+r}\rangle_\infty$ 
can be written as a sum 
over contributions of lines connecting $\sigma_1$ to $\sigma_2$, whose leading term comes from a straight line 
extending between $\sigma_1$ and $\sigma_2$ and subleading terms from zig-zag lines that descend from the leading one 
through addition of decorations. In this way we see that the SCE of 
\begin{e}
\lim_{L_1,L_2\to\infty}\log\Big[\left(\frac{1}{L_2}\log\frac{Z^{(-)}}{Z}\right)\Big/t^{L_1}\Big]\Big/L_1
\ \ \ (=-\rho-\log t)
\end{e}
is \underline{identical}, term by term, to the SCE of 
$\displaystyle\lim_{r\to\infty}\log\big[\langle \sigma_x\sigma_{x+r}\rangle_\infty/t^{r}\big]\big/r\ \ (=-m-\log t)$. 
Hence $m=\rho$. 
\\

The argument above is valid for various other models as long as \textit{on-axis} correlation functions are concerned. 
Then it is natural to ask about \textit{off-axis} correlation functions. (This is the case relevant for the conjecture.) 
From fig.\ref{tri-1} it is easily understood that the leading contribution to the SCE of 
$\displaystyle\frac{1}{L_2}\log\frac{Z^{(-)}}{Z}$ does not come from a single straight line: instead it comes from 
$\begin{pmatrix}L_1\\L_1/2\end{pmatrix}$different loops, all of the same length $L_1$. 
\textit{So we now have a number of different ways to see a given higher-order loop as a sum of any one of the 
leading-order loops and a decoration added to it!} This implies that the counting of diagrams appearing in SCE of 
$\displaystyle\frac{1}{L_2}\log\frac{Z^{(-)}}{Z}$ (and of off-axis correlation function, too) is immensely complicated. 
We even face another problem: since most of the leading-order loops have no translational symmetry, it becomes a 
nontrivial task to show (\ref{prop}). For these reasons we cannot give a mathematically rigorous proof of the conjecture even 
at sufficiently high temperature.

Some caveats are in order.
\begin{itemize}
\item 
First, remember that we \textbf{did} confirm $m=\rho$ for the diagonal correlation function 
in the square Ising model ((\ref{diagonal_CW}) and the accompanying discussion). 
This fact implies that our inability to prove (the very existence of $m,\ \rho$ and) the equality $m=\rho$ 
for off-axis correlation function in SCE approach does not itself constitute a disproof. 

\item 
One may be tempted to argue that, since an exact one-to-one correspondence between the diagrams for SCE of 
$\langle\sigma_1\sigma_2\rangle_\infty$ and those for SCE of $\displaystyle\frac{1}{L_2}\log\frac{Z^{(-)}}{Z}$ 
exists, $m=\rho$ would readily follow even in the off-axis case
\textit{if we presume} the existences of both $m$ and $\rho$. This reasoning is however incorrect, because 
the exact correspondence is present only in the on-axis case. (This point is quite nontrivial.) An example of a loop 
that appears in SCE of $\displaystyle\frac{1}{L_2}\log\frac{Z^{(-)}}{Z}$ as one of the leading contributions 
but has no counterpart among the diagrams in SCE of $\langle\sigma_1\sigma_2\rangle_\infty$ is 
shown in fig.\ref{winding_loop}. (Although somewhat counterintuitive, this loop has a minimum perimeter 
to wind around the lattice in horizontal direction.)

Since we have to send $L_1\to\infty$ \textit{before} $L_2\to\infty$ in estimating $\rho$, those loops can never be 
neglected. From this point of view, it is rather surprising that $m=\rho$ holds in the off-axis case of the square 
Ising model.

\item 
It seems worthwhile to note a qualitative difference between ferromagnetic and antiferromagnetic cases. 
The so-called random path representation $\displaystyle\langle\sigma_1\sigma_2\rangle
=\sum_{\gamma}^{}t^{|\gamma|}$, where the sum runs over every path connecting $\sigma_1$ to $\sigma_2$, loses 
its probabilistic interpretation when $t^{|\gamma|}<0$. This actually happens in a triangular Ising model with 
antiferromagnetic couplings. In such a case it is impossible to apply fertile probability-theoretical methods, so that 
a particular difficulty is envisaged in settling the conjecture when the coupling is antiferromagnetic.
\end{itemize}
Summarizing above, 
we do not have a definitive answer as to the validity of the conjecture even at sufficiently high temperature, 
and further analysis on this topic will be hopefully reported elsewhere.

\begin{figure}[bthp]
\begin{center}
\includegraphics*[width=12.0cm]{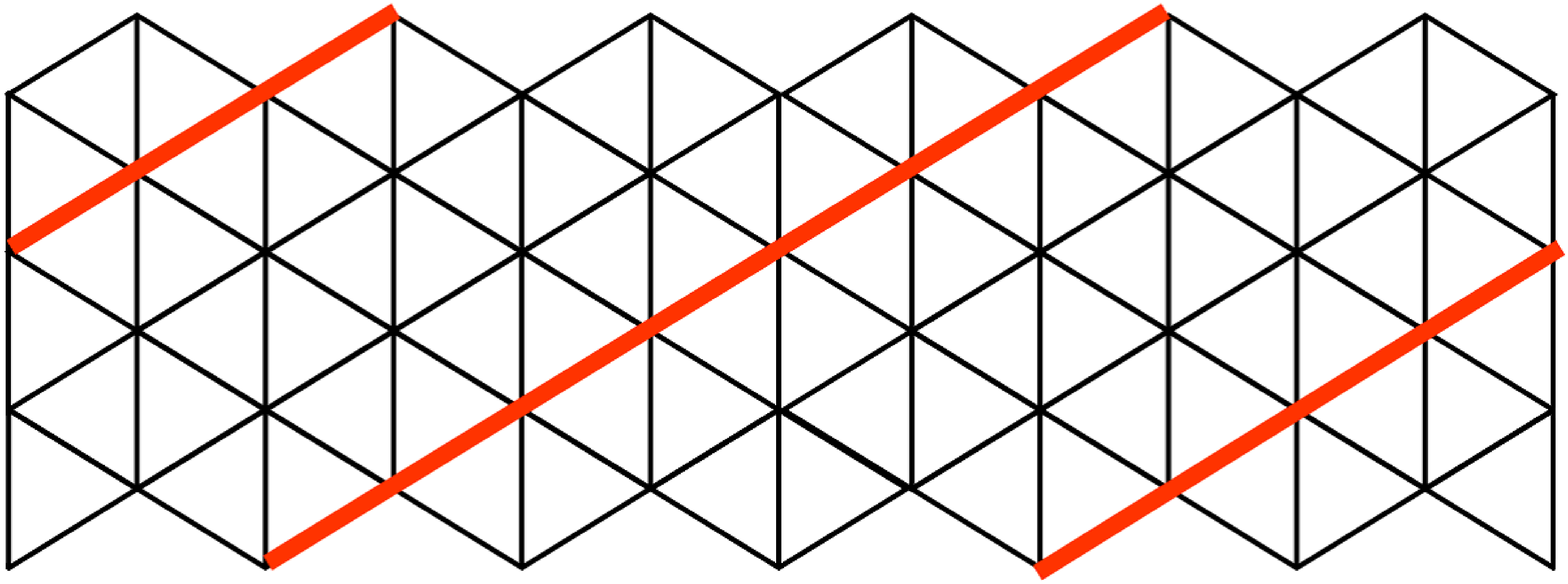}
\end{center}\vspace{-20pt}
\caption{Drawn in red is 
an example of a loop wrapped around the lattice in vertical as well as in horizontal direction. 
Such loops appear numerously if 
the horizontal size of the lattice is much larger than its vertical size.}
\label{winding_loop}
\end{figure}
%
%
%

\section{Conclusion}
In this paper we generalized the inequality of Tomboulis-Yaffe in $SU(2)$ LGT 
to $SU(N)$ LGT and also to general classical spin systems, together 
with a detailed analysis of basic properties of non-Abelian twists. 
Our result is obtained essentially on a finite lattice and gives a rigorous upper 
bound of a Wilson loop and a two-point correlation function. 
An intriguing point is that the inequality obtained for spin models 
does not require the center of the symmetry group, so they can be applied e.g. to $G\times G$ PCM 
with centerless $G$. 
This point seems to be a progress compared with preceding studies in which the center was 
perceived as special without physically convincing motivation\footnote{Some arguments 
that allege the speciality of the center do exist in the literature, but they seem to be rather subtle, 
as discussed attentively in ref.\cite{Borisenko-Skala}.}. 

Furthermore we demonstrated our result explicitly in some solvable models and found 
in the square Ising model that the obtained lower bound of the mass gap is equal to the exact one. 
We also calculated the off-axis mass gap in the triangular Ising model for various couplings, but this time 
the exact mass gap is not known and direct comparison is impossible. 
We conjectured that the bound is indeed saturated and pointed out 
that the conjecture implies the non-analyticity of the mass gap. 
We have tested its validity in several ways, including strong coupling analysis, but a definitive conclusion 
is still lacking and is left for future work.

At present the mechanisms of the quark confinement in 
non-Abelian gauge theories and the mass gap generation in non-Abelian 
spin models still remain elusive, and we hope that our result will be useful for 
further clarification of the issue.
\\
\\
\textit{\large Note added}

After this work was completed, we learned that C. Borgs and E. Seiler had already 
obtained a result very similar to theorem \ref{thm_N} of this paper; 
see Lemma II.8 and the accompanying discussion in ref.\cite{Borgs-Seiler}\footnote{
I thank E. Seiler for kindly pointing out this fact to me.}. 
But since it links the Polyakov loop correlator and not the Wilson loop 
with the electric flux free energy, it is not quite the same as ours. 
It does, however, already imply the 't Hooft's string tension is less than or equal to Wilson's 
(see (II.48) and (II. 50) of ref.\cite{Borgs-Seiler}). Finally we note that their results hardly 
overlap with ours in section \ref{7899}.

\section*{Acknowledgment}
The author thanks Tetsuo Hatsuda, Yoshio Kikukawa, Seiji Miyashita, Masao Ogata, 
Shoichi Sasaki, Hiroshi Suzuki, Shun Uchino and Tamiaki Yoneya for enlightening discussions 
and Ming-Chya Wu for valuable correspondence concerning ref.\cite{Wu-Hu}. 
Thanks also go to the anonymous referee for useful suggestions. 
This work was supported in part by Global COE Program 
``the Physical Sciences Frontier'', MEXT, Japan.

\appendix
\def\thesection{Appendix\ \Alph{section}}
\section{Proof of Lemma \ref{MNlimit}}
For $T>T_c$, (\ref{/163}) and (\ref{/164}) yield
\begin{e}
\jfrac{Z_\Lambda-Z_\Lambda^{(-)}}{Z_\Lambda+Z_\Lambda^{(-)}}=
\jfrac{\Omega_{0,\frac{1}{2}}-\Omega_{0,0}}{\Omega_{\frac{1}{2},\frac{1}{2}}+\Omega_{\frac{1}{2},0}}. \label{x-0}
\end{e}
Let us define $\theta_A^B(\mu,p,M)\geq 0$ by
\begin{e}
\cosh\theta_A^B(\mu,p,M)\equiv\Big|1-B\cos\jfrac{2\pi(p+\mu)}{M}\Big|\Big/\Big|2A\cos\jfrac{\pi(p+\mu)}{M}\Big|.
\end{e}
It is tedious but straightforward to show that the minimum of $\cosh\theta_A^B$ as a function of 
$-1\leq \cos\jfrac{\pi(p+\mu)}{M}\leq 1$ is given by $\displaystyle\frac{g(B)}{|A|}$ (see (\ref{definition_of_g})), 
and that $\displaystyle\frac{g(B)}{|A|}\geq 1$ for every $(t,t_1)\in (0,1)^2$, 
with equality on the critical line. After elementary 
calculations, we find
\begin{e}
(\Omega_{\mu\nu})^2\simeq\Big(\jfrac{A_0}{2}\Big)^{MN}\left\{\prod_{p=0}^{M-1}
\Big|2A\cos\jfrac{\pi(p+\mu)}{M}\Big|\exp\theta_A^B(\mu,p,M)
\right\}^N   \label{gomen}
\end{e}
for $N\gg 1$. 
Since (\ref{gomen}) has no dependence on $\nu$,
\begin{e}
\jlim{N\to\infty}\jfrac{\Omega_{\frac{1}{2},0}}{\Omega_{\frac{1}{2},\frac{1}{2}}}=
\jlim{N\to\infty}\jfrac{\Omega_{0,0}}{\Omega_{0,\frac{1}{2}}}=1.   \label{x-1}
\end{e}
Next, using (\ref{gomen}) we get
\begin{e}
\jfrac{\Omega_{\frac{1}{2},\frac{1}{2}}}{\Omega_{0,\frac{1}{2}}}
\simeq \exp\Big\{N\jsum{k=0}{2M-1}(-1)^{k+1}f_A^B\Big(\jfrac{k}{2M}\Big)\Big\}
\hspace{30pt}\textrm{for\ }N\gg 1,
\label{x-2}
\end{e}
with
\begin{e}
f_A^B(x)\equiv
\jfrac{1}{2}\log\Big\{1-B\cos(2\pi x)
+\sqrt{\big(1-B\cos(2\pi x)\big)^2-\big(2A\cos(\pi x)\big)^2}\,\Big\}.
\end{e}
Let us investigate how fast $\jfrac{\Omega_{0,0}}{\Omega_{0,\frac{1}{2}}}$ converges to 1. 
Using (\ref{gomen}) we can show
\begin{e}
\Big(\jfrac{\Omega_{0,0}}{\Omega_{0,\frac{1}{2}}}\Big)^2
\simeq\sideset{_{}^{}}{_{}^{}}\prod_{p=0}^{M-1}\Big\{1-4\exp(-N\theta_A^B)\Big\}
\hspace{40pt}\textrm{for\ \,}N\gg 1.      \label{usotsuki}
\end{e}
Define $\overline\theta$ as the smallest of $\{\theta_A^B(0,p,M)\}_{p}$. 
Then (\ref{usotsuki}) simplifies to
\begin{e}
\Big(\jfrac{\Omega_{0,0}}{\Omega_{0,\frac{1}{2}}}\Big)^2\simeq 
1-K\,\ee^{-N\overline\theta}\hspace{30pt}\textrm{for\ \,}N\gg 1. \label{x-3}
\end{e}
$K$ is an integer $\in\{4,8,12,16\}$, dependent on $A,\,B$ and $M$. 
Substitution of (\ref{x-1}), (\ref{x-2}) and (\ref{x-3}) into (\ref{x-0}) yields
\begin{e}
1-\jfrac{Z_\Lambda^{(-)}}{Z_\Lambda}\simeq 2\exp\Big\{-N\Big[\overline\theta+\jsum{k=0}{2M-1}(-1)^{k+1}
f_A^B\Big(\jfrac{k}{2M}\Big)\Big]\Big\}\hspace{30pt}\textrm{for\ \,}N\gg 1.   \label{185/*}
\end{e}
Since $\displaystyle\overline\theta=\cosh^{-1}\left(\frac{g(B)}{|A|}\right)+O\Big(\jfrac{1}{M}\Big)$ and 
$\jsum{k=0}{2M-1}(-1)^{k+1}f_A^B\Big(\jfrac{k}{2M}\Big)=O\Big(\jfrac{1}{M}\Big)$ for $M\gg 1$, we find
\begin{gather}
\lim_{M\to\infty}\lim_{N\to\infty}\left(1-\jfrac{Z_\Lambda^{(-)}}{Z_\Lambda}\right)^{1/N}
=\ee^{-\rho},	\label{nobitann}
\\
\rho\equiv\cosh^{-1}\left(\frac{g(B)}{|A|}\right),
\label{nobitannn}
\end{gather}
which is the desired result.
\hfill$\square$

%
%
%
%
%
%
%

\end{document}